\documentclass{article}%
\usepackage{graphicx}
\usepackage{amsmath}
\usepackage{amsfonts}
\usepackage{amssymb}%
\providecommand{\U}[1]{\protect\rule{.1in}{.1in}}
\newtheorem{theorem}{Theorem}

\newtheorem{corollary}[theorem]{Corollary}

\newtheorem{lemma}[theorem]{Lemma}

\newtheorem{proposition}[theorem]{Proposition}
\newtheorem{remark}[theorem]{Remark}

\newenvironment{proof}[1][Proof]{\textbf{#1.} }{\ \rule{0.5em}{0.5em}}
\begin{document}

\title{\textbf{Locality and universality in gravitational anomaly cancellation}}
\date{}
\author{\textsc{Roberto Ferreiro P\'{e}rez}\\Departamento de Econom\'{\i}a Financiera y Actuarial y Estad\'{\i}stica\\Facultad de Ciencias Econ\'omicas y Empresariales, UCM\\Campus de Somosaguas, 28223-Pozuelo de Alarc\'on, Spain\\\emph{E-mail:} \texttt{roferreiro@ccee.ucm.es}}
\maketitle

\begin{abstract}
We obtain necesary and sufficient conditions for gravitational anomaly
cancellation. We show that perturbative gravitational anomalies can never be
cancelled. In a similar way, in dimensions $n\neq3\operatorname{mod}4$ it is
impossible to cancell global anomalies. However, in dimensions
$n=3\operatorname{mod}4$ global anomalies can be cancelled. We prove that the
unique way to cancel the anomaly is by using a Chern-Simons counterterm.
Furthermore, the relationship between the problems of locality and
universality is analyzed for gravitational anomalies.

\end{abstract}

\bigskip


\medskip


\medskip

\noindent\emph{Acknowledgments:\/} Supported by \textquotedblleft Proyecto de
investigaci\'{o}n Santander-UCM PR26/16-20305\textquotedblright.

\section{Introduction}

Although the study of gravitational anomalies is a classical subject\ in
Quantum Field Theory, there are still some open problems related to them (e.g.
see \cite{Monnier} and references therein). But the interest in gravitational
anomalies has increased recently due to the fact that the concept of
gravitational anomaly cancellation has been questioned in \cite{Witten2016}
and \cite{WittenFermionic}. To clarify this concept, a detailed study of
gravitational anomalies is needed. In this paper we obtain necessary and
sufficient conditions for gravitational anomaly cancellation compatible with
locality. We also study the relationship between locality and another problem
in anomaly cancellation, the universality problem.

Roughly speaking, the universality problem for gravitational anomalies
means\ that the conditions for anomaly cancellation should be independent of
the global details of the space-time manifold $M$ of the theory (e.g. see
\cite{BCR},\cite{BCRS}). Furthermore, in \cite{Witten2016}\ Witten
conjectured\ that universality should be generalized to a stronger condition.
Precisely, to have an anomaly free theory consistent with the principles of
unitarity, locality and cutting and pasting, the conditions for anomaly
cancellation (expressed in terms of the Atiyah-Patodi-Singer eta invariant)
should be satisfied for all manifolds of dimension $n+1$ and not just for the
mapping tori of $M$.

We show\ that for perturbative gravitational anomalies universality is just a
consequence of locality. To prove this result we use a characterization of
anomaly cancellation in terms of local cohomology (see bellow) and the fact
that the spaces of diffeomorphisms-invariant local forms are isomorphic for
all compact manifolds of the same dimension. We study if there exists an
analogous relationship between locality and universality for global
gravitational anomalies. Although locality imposes restrictions on the
possible ways to cancel the anomaly, it is not sufficient to imply
universality in the global case and hence it should be imposed. We show that
when universality is imposed, we obtain a generalization of Witten condition.

Now we consider the other basic problem in anomaly cancellation, the locality
problem. This problem is analyzed in \cite{anomalies} for perturbative
anomalies and in \cite{AnomaliesG}\ for global anomalies. We apply the results
of \cite{AnomaliesG} to the case of gravitational anomalies. Let $M$ be
compact oriented manifold of dimension $n$ and let $\mathfrak{Met}M$ be the
space of Riemannian metrics on $M$. We consider the action on $\mathfrak{Met}%
M$ of the group of orientation preserving diffeomorphisms $\mathcal{D}_{M}%
^{+}$. If $\{D_{g}:g\in\mathfrak{Met}M\}$ is a $\mathcal{D}_{M}^{+}%
$-equivariant family of elliptic operators acting on fermionic fields $\psi$
and parametrized by $\mathfrak{Met}M$, then the Lagrangian density
$\lambda_{D}(\psi,g)=\bar{\psi}iD_{g}\psi$ is $\mathcal{D}_{M}^{+}$-invariant.
The theory can be anomalous because the corresponding partition function
$\mathcal{Z}(g)=\int\mathcal{D}\psi\mathcal{D}\bar{\psi}\exp\left(  -\int
_{M}\bar{\psi}iD_{g}\psi\right)  $ could fail to be $\mathcal{D}_{M}^{+}%
$-invariant. It can be seen that\ the modulus of $\mathcal{Z}(s)$ is
$\mathcal{D}_{M}^{+}$-invariant. Hence we have $\mathcal{Z}(\phi\cdot
g)=\mathcal{Z}(g)\cdot\exp(2\pi i\cdot\alpha_{\phi}(g))$ where $\alpha
\colon\mathcal{D}_{M}^{+}\times\mathfrak{Met}M\rightarrow\mathbb{R}%
/\mathbb{Z}$ satisfies the cocycle condition $\alpha_{\phi_{2}\phi_{1}%
}(g)=\alpha_{\phi_{1}}(g)+\alpha_{\phi_{2}}(\phi_{1}g)$. The anomaly can be
cancelled if there exists a \emph{local} functional $\Lambda\in\Omega
_{\mathrm{loc}}^{0}(\mathfrak{Met}M)$ satisfying the condition%
\begin{equation}
\alpha_{\phi}(g)=\Lambda(\phi\cdot g)-\Lambda(g). \label{Lambda}%
\end{equation}
We recall that local functionals are those of the form $\Lambda=\int
_{M}\lambda$ for a lagrangian density $\lambda$. If the condition
(\ref{Lambda}) is satisfied, we can redefine the lagrangian density to
$\lambda^{\prime}(g)=\bar{\psi}iD_{g}\psi-\lambda(g)$ and the new partition
function $\mathcal{Z}^{\prime}=\mathcal{Z}\cdot\exp(-2\pi i\Lambda)$ is
$\mathcal{D}_{M}^{+}$-invariant. Alternatively, we can consider $\lambda$ as
the effective lagrangian of the theory. Note that in order to cancel the
anomaly $\exp(2\pi i\Lambda)$ should not be $\mathcal{D}_{M}^{+}$-invariant.

Our objective in this paper is to obtain necessary and sufficient conditions
for anomaly cancellation, and hence we need to analyze the condition
(\ref{Lambda}) in detail. Following \cite{AnomaliesG} we say that the
topological anomaly cancels if condition (\ref{Lambda})\ is satisfied for a
functional $\Lambda\in\Omega^{0}(\mathfrak{Met}M)$, and that the physical
anomaly cancels if the condition (\ref{Lambda}) is satisfied for a local
functional $\Lambda\in\Omega_{\mathrm{loc}}^{0}(\mathfrak{Met}M)$.
Furthermore, if the condition (\ref{Lambda}) is satisfied only for the
connected component of the identity $\mathcal{D}_{M}^{0}$ in $\mathcal{D}%
_{M}^{+}$ we say that the perturbative\ (or local) anomaly cancels. If it is
satisfied for all the elements of $\mathcal{D}_{M}^{+}$ we say that the global
anomaly cancels.

Topological anomaly cancellation admits a geometrical interpretation in terms
line bundles (e.g. see \cite{AS}, \cite{Monnier},\cite{AnomaliesG}). The
cocycle $\alpha$ determines an action on the trivial bundle $\mathcal{L}%
=\mathfrak{Met}M\times\mathbb{C}\rightarrow\mathfrak{Met}M$ by setting
$\phi_{\mathcal{L}}(g,u)=(\phi(g),u\cdot\exp(2\pi i\alpha_{\phi}(g)))$ for
$g\in\mathfrak{Met}M$ and $u\in\mathbb{C}$, and $\mathcal{Z}$ is a
$\mathcal{D}_{M}^{+}$-equivariant section of $\mathcal{L}$. We can also
consider the principal $U(1)$-bundle $\mathcal{U}=\mathfrak{Met}M\times
U(1)\rightarrow\mathfrak{Met}M$. The $\mathcal{D}_{M}^{+}$-equivariant
$U(1)$-bundle $\mathcal{U}\rightarrow\mathfrak{Met}M$ is called the anomaly
bundle and admits a natural $\mathcal{D}_{M}^{+}$-invariant connection $\Xi$.
Usually the partition function is defined in terms of determinants of elliptic
operators and $\mathcal{L}$ can be identified with a determinant or Pffafian
line bundle and $\Xi$ with the Bismut-Freed connection (see \cite{BF1}).
Furthermore, if $\Lambda\in\Omega^{0}(\mathfrak{Met}M)$ satisfies condition
(\ref{Lambda}) then $\exp(2\pi i\Lambda)$ determines a $\mathcal{D}_{M}^{+}%
$-equivariant section of of $\mathcal{U}$. Hence the condition for topological
anomaly cancellation is equivalent to the fact that the anomaly bundle
$\mathcal{U}$ is a trivial $\mathcal{D}_{M}^{+}$-equivariant $U(1)$-bundle. It
is easy to obtain necessary and sufficient conditions for topological anomaly
cancellation in\ terms of the equivariant curvature and holonomy of the
connection $\Xi$ (see \cite{AnomaliesG} and Section
\ref{SectionLocAnomalyCanc}\ for details). It follows that the anomaly can be
cancelled by using a $\mathcal{D}_{M}^{+}$-invariant$\ 1$-form $\beta\in
\Omega^{1}(\mathfrak{Met}M)$.

However, the physical anomaly cancellation requires that $\Lambda$ should be a
local functional $\Lambda\in\Omega_{\mathrm{loc}}^{0}(\mathfrak{Met}M)$, and
hence $\exp(2\pi i\Lambda)$ should be a special type of section of the anomaly
bundle, a local section. The notion of local section and the conditions for
physical anomaly cancellation require an adequate definition of local form
$\Omega_{\mathrm{loc}}^{k}(\mathfrak{Met}M)$ for $k>0$. We recall that the
formulation of variational calculus in terms of jet bundles provides a
generalization of the concept of local functionals to higher order forms (see
\cite{localVB} and Section \ref{Seclocalforms}\ for details). Furthermore, the
Atiyah-Singer Index theorem for families implies that the curvature of $\Xi$
is a local form, $\mathrm{curv}(\Xi)\in\Omega_{\mathrm{loc}}^{2}%
(\mathfrak{Met}M)$. We show that the necessary and sufficient condition for
perturbative physical anomaly cancellation is the existence of a \emph{local
}$\mathcal{D}_{M}^{+}$-invariant $1$-form $\beta\in\Omega_{\mathrm{loc}}%
^{1}(\mathfrak{Met}M)$ such that $\mathrm{curv}(\Xi)=d\beta$. Hence the
perturbative anomaly can be represented as a cohomology class on the local
$\mathcal{D}_{M}^{+}$-invariant cohomology $H_{\mathrm{loc}}^{2}%
(\mathfrak{Met}M)^{\mathcal{D}_{M}^{+}}$. It is important to note that the
local forms and local cohomology of $\mathfrak{Met}M$ are very different to
the ordinary ones, and this has implications in the possibilities for anomaly
cancellation. We prove that $H_{\mathrm{loc}}^{2}(\mathfrak{Met}%
M)^{\mathcal{D}_{M}^{+}}$\ is independent of the manifold $M$ and depends only
on the dimension of $M$. Hence if this condition is satisfied for a manifold
$M$, it is also satisfied for any other manifold of the same dimension. In
this way, it follows that for perturbative anomalies universality is a
consequence of locality as commented below.

If the perturbative anomaly cancels, we can still have global anomalies. In
this case the anomaly bundle admits a $\mathcal{D}_{M}^{+}$-basic connection
and the possible forms that can be used to cancel the anomaly are determined
by $H_{\mathrm{loc}}^{1}(\mathfrak{Met}M)^{\mathcal{D}_{M}^{+}}$. We show that
$H_{\mathrm{loc}}^{1}(\mathfrak{Met}M)^{\mathcal{D}_{M}^{+}}$ is generated by
the exterior differentials of the Chern-Simons actions associated to
polynomials $p\in I^{(n+1)/2}(O(n))$. In particular, in dimensions
$n\neq3\operatorname{mod}4$ it is impossible to cancel the anomaly. However,
in dimensions $n=3\operatorname{mod}4$, it could be possible to cancel the
anomaly with $\Lambda$ the Chern-Simons action associated to $p$. For
determinant bundles we use Witten formula for the holonomy of the Bismut-Freed
connection to obtain necessary and sufficient conditions for anomaly
cancellation. If $M_{\phi}$ denotes the mapping torus of $\phi\in
\mathcal{D}_{M}^{+}$ and $\eta_{D}$ the Atiyah-Patodi-Singer eta invariant the
condition is
\begin{equation}
\tfrac{1}{4}\eta_{D}(M_{\phi})=p(M_{\phi})\operatorname{mod}\mathbb{Z}\text{
for any }\phi\in\mathcal{D}_{M}^{+}. \label{AnomalyIntro}%
\end{equation}
The preceding condition is not universal in the sense that it can be satisfied
for a particular manifold $M$ but not for all manifolds of the same dimension.
A universal version of the condition (\ref{AnomalyIntro}) is obtained by
requiring \ that this condition should\ be satisfied for any oriented manifold
$N$\ of dimension $n+1$ and not only for mapping tori. This provides a
condition weaker that the condition $\tfrac{1}{4}\eta_{D}%
(N)=0\operatorname{mod}\mathbb{Z}$ introduced in \cite{WittenFermionic}. For
example, we show that for Majorana fermions in dimension 3 the anomaly can be
cancelled in this more general sense but not in Witten sense.

We extend our result to include orientation reversing diffeomorphisms. The
problem in this case is that $M_{\phi}$ is unorientable and we cannot compute
$p(M_{\phi})$. We show that the condition in this case is $\tfrac{1}{4}%
\eta_{D}(M_{\phi})=\tfrac{1}{2}p(M_{\phi^{2}})\operatorname{mod}\mathbb{Z}$
for any $\phi\in\mathcal{D}_{M}$. Furthermore the universal version of this
condition is $\tfrac{1}{4}\eta_{D}(N)=\tfrac{1}{2}p(\widetilde{N}%
)\operatorname{mod}\mathbb{Z}$ where $\widetilde{N}\rightarrow N$ is the
double cover of $N$.

In Section \ref{SectionCS} we define the gravitational\ Chern-Simons action
and we show how they can be used to cancel gravitational anomalies. In the
rest of the paper we introduce the concepts that are necessary to prove that
the unique possible way to cancel the anomaly is by using a Chern-Simons
action. In Section \ref{Seclocalforms}\ \ we study the local cohomology of
$\mathfrak{Met}M$\ and in Section \ref{SectionLocAnomalyCanc}\ we apply the
results of \cite{AnomaliesG} to characterize gravitational anomaly
cancellation. We also study the universal generalization of this condition an
its relation to the condition considered by Witten in \cite{WittenFermionic}.

\section{Chern-Simons actions\label{SectionCS}}

Let $M$ be an oriented compact manifold of dimension $n=3\operatorname{mod}4$
and $p\in I^{(n+1)/2}(O(n))$ an invariant polynomial. We recall that
$I^{\bullet}(O(n))$ is generated over $\mathbb{R}$ by the Pontryagin
polynomials $p_{i}$, and also by the trace of the even powers $T^{2k}%
(A)=\mathrm{tr}(A^{2k})$ for $A\in\mathfrak{so}(n)$. Hence we can consider
$I^{\bullet}(O(n))\subset I^{\bullet}(Gl(n,\mathbb{R}))$.

We denote by $\mathfrak{Met}M$ the space of Riemannian metrics on $M$. The
group $\mathcal{D}_{M}$ of diffeomorphisms of $M$ acts on $\mathfrak{Met}M$ by
setting $\phi\cdot g=\phi_{\mathfrak{Met}M}(g)=(\phi^{-1})^{\ast}g$ for
$g\in\mathfrak{Met}M$ and $\phi\in\mathcal{D}_{M}$. The Levi-Civita connection
of $g\in\mathfrak{Met}M$, considered as a connection on the frame bundle
$FM\rightarrow M$, is denoted by $\omega^{g}$. If $A_{0}$ is a connection on
$FM$, we define the Chern-Simons action with background connection $A_{0}$ as
the function $CS_{p,A_{0}}\colon\mathfrak{Met}M\rightarrow\mathbb{R}$ defined
by $CS_{p,A_{0}}(g)=\int_{M}Tp(\omega^{g},A_{0})$ for $g\in\mathfrak{Met}M$.
If we change the background connection to $A_{0}^{\prime}$, then we have
$CS_{p,A_{0}^{\prime}}=CS_{p,A_{0}}+\int_{M}$ $Tp(A_{0},A_{0}^{\prime})$.
Hence the Chern-Simons action is independent of $A_{0}$ up to the constant
$\int_{M}Tp(A_{0},A_{0}^{\prime})$. For integral polynomials it is possible to
fix this constant $\operatorname{mod}\mathbb{Z}$ in a canonical way (see
Remark \ref{Remark IntegerCS}).

\begin{remark}
\emph{The Chern-Simons action is usually defined only for paralelizable
manifolds (for example }$3$\emph{-manifolds) by using a trivialization of the
frame bundle }$FM\rightarrow M$\emph{. This is a particular case of our
definition, where we take }$A_{0}$\emph{ to be the connection induced by the
trivialization.}
\end{remark}

As we want to use the Chern-Simons actions to cancel gravitational anomalies,
we need to study its variation under diffeomorphisms. If $\phi\in
\mathcal{D}_{M}$ and $g\in\mathfrak{Met}M$ we define $\delta_{\phi}%
^{p}(g)=CS_{p,A_{0}}(\phi\cdot g)-CS_{p,A_{0}}(g)$ (clearly it does not depend
on the background connection $A_{0}$). Furthermore we also have

\begin{proposition}
\label{Independence}If $\phi\in\mathcal{D}_{M}^{+}$ then $\delta_{\phi}%
^{p}(g)$ is independent of $g$.
\end{proposition}

\begin{proof}
We have
\begin{align*}
\delta_{\phi}^{p}(g)\!-\!\delta_{\phi}^{p}(h)\!  &  =\!%
{\textstyle\int_{M}}
Tp(\omega^{\phi\cdot g},A_{0})\!-\!%
{\textstyle\int_{M}}
Tp(\omega^{g},A_{0})\!-\!%
{\textstyle\int_{M}}
Tp(\omega^{\phi\cdot h},A_{0})\!+\!%
{\textstyle\int_{M}}
Tp(\omega^{h},A_{0})\\
&  =%
{\textstyle\int_{M}}
Tp(\omega^{\phi\cdot g},\omega^{\phi\cdot h})-%
{\textstyle\int_{M}}
Tp(\omega^{g},\omega^{h})\\
&  =%
{\textstyle\int_{M}}
(\phi^{-1})^{\ast}Tp(\omega^{g},\omega^{h})-%
{\textstyle\int_{M}}
Tp(\omega^{g},\omega^{h})=0.
\end{align*}

In the last equation we use that $\int_{M}(\phi^{-1})^{\ast}\alpha=\int
_{M}\alpha$ for any $\alpha\in\Omega^{n}(M)$ because $\phi\in D_{M}^{+}$.
\end{proof}

\begin{remark}
\emph{Another way to prove this Proposition is by using that} $d(CS_{p,A_{0}%
})$ \emph{is} $D_{M}^{+}$-\emph{invariant (see Section \ref{Seclocalforms}).
Then we have} $d(\delta_{\phi}^{p})=\phi_{\mathfrak{Met}M}^{\ast}%
d(CS_{p,A_{0}})-d(CS_{p,A_{0}})=0$\emph{, and hence }$\delta_{\phi}^{p}$\emph{
is independent of} $g\in\mathfrak{Met}M$.
\end{remark}

By Proposition \ref{Independence} we can denote $\delta_{\phi}^{p}(g)$ simply
by $\delta_{\phi}^{p}$.

\begin{proposition}
\label{Independence2}If $\phi,\phi^{\prime}\in\mathcal{D}_{M}^{+}$ then
$\delta_{\phi\cdot\phi^{\prime}}^{p}=\delta_{\phi^{\prime}}^{p}+\delta
_{\phi^{\prime}}^{p}$.
\end{proposition}

\begin{proof}
We have%
\begin{align*}
\delta_{\phi\cdot\phi^{\prime}}^{p}(g)  &  =CS_{p,A_{0}}(\phi\cdot\phi
^{\prime}\cdot g)-CS_{p,A_{0}}(g)\\
&  =CS_{p,A_{0}}(\phi\cdot\phi^{\prime}\cdot g)-CS_{p,A_{0}}(\phi^{\prime
}\cdot g)+CS_{p,A_{0}}(\phi^{\prime}\cdot g)-CS_{p,A_{0}}(g)\\
&  =\delta_{\phi}^{p}(\phi^{\prime}\cdot g)+\delta_{\phi^{\prime}}^{p}(g).
\end{align*}

The result follows by using Proposition \ref{Independence}.
\end{proof}

If $\phi\in\mathcal{D}_{M}$ we define its mapping torus $M_{\phi}%
=(M\times\lbrack0,1])/\sim_{\phi}$ where $(x,0)\sim_{\phi}(\phi(x),1)$. When
$\phi\in\mathcal{D}_{M}^{+}$ the orientation on $M$ induces an orientation on
$M_{\phi}$. If $\phi$ reverses the orientation then $M_{\phi}$ is
unorientable. If $p\in I^{(n+1)/2}(O(n))$ and $\phi\in\mathcal{D}_{M}^{+}$\ we
define $p(M_{\phi})\in\mathbb{R}$ as the characteristic number associated to
$p$ on $M_{\phi}$, i.e. $p(M_{\phi})=\int_{M_{\phi}}p(F)$, where $F$ is the
curvature of any connection on $F(M_{\phi})$ (it is independent of the
connection by Chern-Weil theory).

Witten formula express the variation of the partition function in terms of the
eta invariant of an operator on the mapping torus. We obtain a similar result
for the Chern-Simons action

\begin{theorem}
\label{MappingTorus} For any $\phi\in\mathcal{D}_{M}^{+}$ we have
$\delta_{\phi}^{p}=p(M_{\phi})$.
\end{theorem}

\begin{proof}
a) The bundle $FM\times I\rightarrow M\times I$ is a principal
$Gl(n,\mathbb{R})$-bundle. For small $\varepsilon$ we can find a connection
$\overline{\omega}$ on this bundle that coincides with $\omega^{g}$\ on
$FM\times\{t\}$\ for $t<\varepsilon$ and with $\omega^{\phi\cdot g}$ for
$t>1-\varepsilon$. The connection $\overline{\omega}$ induces a connection
$\overline{\omega}_{\phi}$ on the quotient bundle $F^{\phi}M=(FM\times
I\mathbb{)}/\sim_{\phi}\rightarrow M_{\phi}$. We denote by $\overline{\Omega}$
and $\overline{\Omega}_{\phi}$ the curvature forms of $\overline{\omega}$ and
$\overline{\omega}_{\phi}$.

If $\mathrm{pr}_{1}\colon M\times I\rightarrow M$ and $\overline{\mathrm{pr}%
}_{1}\colon FM\times I\rightarrow FM$ are the projections, then we have
$p(\overline{\mathrm{pr}}_{1}^{\ast}F_{0})=\mathrm{pr}_{1}^{\ast}p(F_{0})=0$
by dimensional reasons. As $M_{\phi}$ is an oriented manifold and we can
consider the integral
\begin{align*}%
{\textstyle\int_{M_{\phi}}}
p(\overline{\Omega}_{\phi})  &  =%
{\textstyle\int_{M\times I}}
p(\overline{\Omega})=%
{\textstyle\int_{M\times I}}
d(Tp(\overline{\omega},\overline{\mathrm{pr}}_{1}^{\ast}A_{0}))\\
&  =%
{\textstyle\int_{M\times\{1\}}}
Tp(\overline{\omega},\overline{\mathrm{pr}}_{1}^{\ast}A_{0})-%
{\textstyle\int_{M\times\{0\}}}
Tp(\overline{\omega},\overline{\mathrm{pr}}_{1}^{\ast}A_{0})\\
&  =%
{\textstyle\int_{M}}
Tp(\omega^{\phi\cdot g},A_{0})-%
{\textstyle\int_{M}}
Tp(\omega^{g},A_{0})=\phi_{\mathfrak{Met}M}^{\ast}(CS_{p,A_{0}})-CS_{p,A_{0}}%
\end{align*}

We need to prove that the characteristic classes of the bundle $F^{\phi}%
M\ $coincide with the characteristic classes of$\ M_{\phi}$. As $I^{\bullet
}(O(n))$ is generated over $\mathbb{R}$ by the Pontryagin polynomials, it is
sufficient to prove the result for the real Pontryagin classes. The vertical
bundle $V(M_{\phi})\rightarrow M_{\phi}$ of the fibration $q\colon M_{\phi
}\rightarrow S^{1}$\ is a vector bundle associated to the principal
$Gl(n,\mathbb{R})$-bundle$\ F^{\phi}M\rightarrow M_{\phi}$. We have
$T(M_{\phi})\simeq V(M_{\phi})\oplus q^{\ast}TS^{1}$, and as $q^{\ast}TS^{1}$
is trivial we have $p_{k}(T(M_{\phi}))=p_{k}(V(M_{\phi}))$.
\end{proof}

The following result follows as a consequence of the preceding Theorem. It can
also be obtained from the existence of equivariant Pontryagin forms (see
Section \ref{Seclocalforms})

\begin{corollary}
\label{CSD0}If $\phi\in\mathcal{D}_{M}^{0}$ then $\delta_{\phi}^{p}=0$.
\end{corollary}

\begin{proof}
It follows from the fact that if $\phi_{t}\subset\mathcal{D}_{M}^{0}$ is a
curve joining $\mathrm{id}_{M}$ and $\phi$, then the map $FM\times
I\rightarrow F^{\phi}M$, $(u,t)\mapsto\lbrack((\phi_{t})_{\ast}u,t)]$\ is an isomorphism.
\end{proof}

It follows from this Corollary that the Chern-Simons action cannot be used to
cancel perturbative anomalies. However, it can be used to cancel global
anomalies because we can have $\delta_{\phi}^{p}\neq0$ for $\phi\in
\mathcal{D}_{M}^{+}$. Furthermore, we conclude that $\delta^{p}$ can be
considered as a group homomorphism $\delta^{p}\in\mathrm{Hom}(\mathcal{D}%
_{M}^{+}/\mathcal{D}_{M}^{0},\mathbb{R)}$.

\begin{remark}
\label{Remark IntegerCS}\emph{If the polynomial }$p$\emph{ determines an
integral characteristic class, or more generally, if }$p(M_{\phi})$\emph{ is
an integer for any }$\phi\in\mathcal{D}_{M}^{+}$, \emph{then we obtain as a
corollary of Theorem \ref{MappingTorus}\ that }$CS_{p,A_{0}}$\emph{ is
}$\mathcal{D}_{M}^{+}$\emph{-invariant} $\operatorname{mod}\mathbb{Z}$,
\emph{and hence} $\exp(2\pi i\cdot CS_{p,A_{0}})$ \emph{is }$\mathcal{D}%
_{M}^{+}$\emph{-invariant. }

\emph{Furthermore, if }$p$\emph{ determines an integral characteristic class,
the Chern-Simons theory can be used to give a definition the Chern-Simons
action independent of the background connection }$A_{0}$\emph{ (see
\cite{DW},\cite{flat}). If }$Z_{n}(M)$\emph{ is the space of }$n$\emph{-cycles
on }$M$\emph{ and }$\chi_{A}\colon Z_{n}(M)\rightarrow\mathbb{R}/\mathbb{Z}$
\emph{is the Chern-Simons differential character associated to a connection
}$A$\emph{ on }$FM$\emph{,\ we can define }$CS_{p}\colon\mathfrak{Met}%
M\rightarrow\mathbb{R}/\mathbb{Z}$ \emph{by} $\ CS_{p}(g)=\chi_{\omega^{g}%
}(M)$ \emph{and we have }$CS_{p}(g)=\chi_{A_{0}}(M)+CS_{p,A_{0}}%
(g)\operatorname{mod}\mathbb{Z}$. \emph{In this way we can define in a
canonical way the }$\mathcal{D}_{M}^{+}$\emph{-invariant exponentiated
Chern-Simons action }$\exp(2\pi i\cdot CS_{p})$\emph{. This fact is
interesting to study Chern-Simons theory as a Topological Quantum field
theory, but it cannot be used to cancel anomalies. }

\emph{However, for an arbitrary }$p\in I^{(n+1)/2}(O(n))$\emph{ the
function}$\ CS_{p,A_{0}}$\emph{ is not }$D_{M}^{+}$\emph{-invariant
}$\operatorname{mod}\mathbb{Z}$\emph{, and this fact can be used to cancel
gravitational anomalies if the variation of }$\exp(2\pi i\cdot CS_{p,A_{0}}%
)$\emph{\ cancels the variation of the partition function.}
\end{remark}

If the perturbative anomaly cancels then we have $\mathcal{Z}(\phi\cdot
g)=\mathcal{Z}(g)\exp(2\pi i\cdot\kappa_{\phi})$\ where $\kappa_{\phi}%
\in\mathrm{Hom}(\mathcal{D}_{M}^{+}/\mathcal{D}_{M}^{0},\mathbb{R}%
/\mathbb{Z)}$. We can\ cancel the anomaly with a Chern-Simons counterterm if
there exists $p\in I^{(n+1)/2}(O(n))$ such that $\kappa_{\phi}=\delta_{\phi
}^{p}$ for any $\phi\in\mathcal{D}_{M}^{+}$. Furthermore, we prove in Section
\ref{SectionLocAnomalyCanc}\ \ that this is the unique possible way to cancel
the anomaly. For Dirac operators $\kappa_{\phi}$ has been computed by Witten
(e.g. see \cite{WittenGlobAn},\cite{Witten2016},\cite{WittenFermionic}) and we
have $\kappa_{\phi}=\frac{1}{4}\eta_{D}(M_{\phi})$, where $\eta_{D}$ is the
reduced Atiyah-Patodi-Singer eta invariant of a differential operator on
$M_{\phi}$. If Witten formula applies, then the necessary and sufficient
condition for global anomaly cancellation is the existence of $p\in
I^{(d+1)/2}(O(n))$ such that $\frac{1}{4}\eta_{D}(M_{\phi})=p(M_{\phi
})\operatorname{mod}\mathbb{Z}$ for any $\phi\in\mathcal{D}_{M}^{+}$.

We recall that the eta invariant $\eta_{D}(N)$ is a spectral invariant of the
metric defined by $\eta_{D}=\dim\ker D+\underset{\varepsilon\rightarrow0^{+}%
}{\lim}\sum_{k}\mathrm{sign}(\lambda_{k})\exp(-\varepsilon\lambda_{k}^{2})$,
where $\lambda_{k}$ are the eigenvalues of $D$. It is know\ that in general
$\eta_{D}(N)$ cannot be obtained as the integral of a local form. However,
(e.g. see \cite{Stolz}) for certain operators $\eta_{D}(N)$ can be
obtained\ in terms of characteristic numbers. We study this condition and its
generalizations in more detail in Section \ref{SectionLocAnomalyCanc}.

Note that $\frac{1}{4}\eta_{D}(M_{\phi})$ is independent of the orientation on
$M$, but the sign of $p(M_{\phi})\ $changes if we change the orientation.
Hence, if $\frac{1}{4}\eta_{D}(M_{\phi})=p(M_{\phi})\operatorname{mod}%
\mathbb{Z}$ then we have $p(M_{\phi})=-p(M_{\phi})\operatorname{mod}%
\mathbb{Z}$ and hence $2p(M_{\phi})=0\operatorname{mod}\mathbb{Z}$. We
conclude that we can cancel the anomaly only if $\frac{1}{4}\eta_{D}(M_{\phi
})=0$ or $\frac{1}{2}\operatorname{mod}\mathbb{Z}$. This implies a possible
change of sign in the partition function $\mathcal{Z}(\phi\cdot g)=\pm
\mathcal{Z}(g)$. We recall (see \cite{WittenFermionic})\ that this sign
anomaly is the kind of anomaly that can appear for real or pseudoreal fermions.

\subsection{Orientation reversing diffeomorphisms\label{CSnoOrient}}

If $\phi$ is an orientation reversing diffeomorphisms, then the results of the
previous section can not be applied. Proposition \ref{Independence} does not
hold for $\phi$ and $M_{\phi}$ is unorientable, and hence we cannot compute
$p(M_{\phi})$. However, orientation reversing diffeomorphisms are important
symmetries and they should be consider in anomaly cancellation (e.g. see
\cite{WittenFermionic}). Furthermore, in \cite{WittenFermionic} and
\cite{WittenParity} it is analyzed also the case in which $M$ is unorientable.

One way to solve this problem is to consider the orientation as an independent
field. We define $\widetilde{\mathfrak{Met}}M=\{(g,\mathfrak{o}):$
$g\in\mathfrak{Met}M$ and $\mathfrak{o}$ is an orientation on $M\}$ and
$\mathcal{D}_{M}$ acts in a natural way on $\widetilde{\mathfrak{Met}}M$. We
extend the Chern-Simons action to $\widetilde{\mathfrak{Met}}M$ by setting
$CS_{p,A_{0}}(g,\mathfrak{o})=\int_{(M,\mathfrak{o})}Tp(\omega^{g},A_{0})$ and
then Propositions \ref{Independence} and \ref{Independence2} are valid for
arbitrary diffeomorphisms. Furthermore, for any $\phi\in\mathcal{D}_{M}$ we
have $\delta_{\phi}^{p}=\frac{1}{2}\delta_{\phi^{2}}^{p}=\frac{1}{2}%
p(M_{\phi^{2}})$ as $\phi^{2}\in\mathcal{D}_{M}^{+}$. If we combine this
result with Witten formula we conclude that the $\mathcal{D}_{M}$-anomaly
cancels if $\tfrac{1}{4}\eta_{D}(M_{\phi})=\frac{1}{2}p(M_{\phi^{2}%
})\operatorname{mod}\mathbb{Z}$ for any $\phi\in\mathcal{D}_{M}$.

Again the left side is independent of the orientation on $M$ and the right
side depends on it. Hence, also for orientation reversing diffeomorphisms we
can cancel only a sign anomaly.

\section{Equivariant cohomology in the Cartan model\label{SectionEquiCoho}}

We recall the definition of equivariant cohomology in the Cartan model
(\emph{e.g. }see \cite{BGV,GS}). Suppose that we have a left action of a
connected Lie group $\mathcal{G}$ on a manifold $N$. We denote by $\Omega
^{k}(N)^{\mathcal{G}}$the space of $\mathcal{G}$-invariant forms on $N$ and by
$H^{k}(N)^{\mathcal{G}}$ the $\mathcal{G}$-invariant cohomology of $N$. Let
$\Omega_{\mathcal{G}}^{\bullet}(N)=\mathcal{P}^{\bullet}(\mathrm{Lie\,}%
\mathcal{G},\Omega^{\bullet}(N))^{\mathcal{G}}$ be the space of $\mathcal{G}%
$-invariant polynomials on $\mathrm{Lie\,}\mathcal{G}$ with values in
$\Omega^{\bullet}(N)$ with the graduation $\deg(\alpha)=2k+r$ if $\alpha
\in\mathcal{P}^{k}(\mathrm{Lie\,}\mathcal{G},\Omega^{r}(N))$. Let
$D\colon\Omega_{\mathcal{G}}^{q}(N)\rightarrow\Omega_{\mathcal{G}}^{q+1}(N)$
be the Cartan differential, $(D\alpha)(X)=d(\alpha(X))-\iota_{X_{N}}\alpha
(X)$, $X\in\mathrm{Lie\,}\mathcal{G}$. On $\Omega_{\mathcal{G}}^{\bullet}(N)$
we have $D^{2}=0$, and the equivariant cohomology (in the Cartan model) of $N$
with respect to the action of $\mathcal{G}$ is defined as the cohomology of
this complex.

A $\mathcal{G}$-equivariant $1$-form $\alpha\in\Omega_{\mathcal{G}}%
^{1}(\mathcal{N})$ is just a $\mathcal{G}$-invariant $1$-form $\alpha\in
\Omega^{1}(\mathcal{N})^{\mathcal{G}}$. It is $D$-closed if and only if it is
$\mathcal{G}$-basic, i.e., if $d\alpha=0$ and $\iota_{X_{N}}\alpha=0$ for any
$X\in\mathrm{Lie\,}\mathcal{G}$.

If $\varpi\in\Omega_{\mathcal{G}}^{2}(N)$ is a $\mathcal{G}$-equivariant
$2$-form, then we have $\varpi=\omega+\mu$ where $\omega\in\Omega^{2}(N)$ is
$\mathcal{G}$-invariant and $\mu\in\mathrm{Hom}\left(  \mathrm{Lie\,}%
\mathcal{G},\Omega^{0}(N)\right)  ^{\mathcal{G}}$. In particular we have
$\mu([Y,X])=L_{Y_{N}}\left(  \mu(X)\right)  $ for any $X,Y\in\mathrm{Lie\,}%
\mathcal{G}$. The map $\alpha\mapsto\alpha(0)$ induces a map $H_{\mathcal{G}%
}^{k}(N)\rightarrow H^{k}(N)^{\mathcal{G}}$. We have the following

\begin{proposition}
\label{LieAlgCoho}If $H^{1}(\mathrm{Lie\,}\mathcal{G})=0$ then the map
$H_{\mathcal{G}}^{2}(N)\rightarrow H^{2}(N)^{\mathcal{G}}$ is injective.
\end{proposition}

\begin{proof}
Let $\varpi=\omega+\mu\in\Omega_{\mathcal{G}}^{2}(N)$ be a closed
$\mathcal{G}$-equivariant $2$-form, and assume that there exists $\beta
\in\Omega^{1}(N)^{\mathcal{G}}$ such that $d\beta=\omega$. Then $\tau
=\varpi-D\beta\in\Omega_{\mathcal{G}}^{2}(N)$ is $D$-closed and $\tau
\in\mathrm{Hom}\left(  \mathrm{Lie\,}\mathcal{G},\Omega^{0}(N)\right)
^{\mathcal{G}}$. But $D\tau=0$ implies $d(\tau(X))=0$ for any $X\in
\mathrm{Lie\,}\mathcal{G}$ and hence $\tau\in\mathrm{Hom}\left(
\mathrm{Lie\,}\mathcal{G},\mathbb{R}\right)  ^{\mathcal{G}}$. Furthermore, by
the $\mathcal{G}$-invariance of $\tau$ we have$\ \tau([Y,X])=0$ for any
$X,Y\in\mathrm{Lie\,}\mathcal{G}$. As $H^{1}(\mathrm{Lie\,}\mathcal{G})=0$ we
conclude that $\tau=0$, and hence $\varpi=D\beta$.
\end{proof}

If $\omega\in\Omega^{k}(N)^{\mathcal{G}}$ is closed, we say that $\varpi
\in\Omega_{\mathcal{G}}^{k}(N)$ is a $\mathcal{G}$-equivariant extension of
$\omega$ if $\varpi(0)=\omega$ and $D\omega=0$. In general there can be
obstructions to the existence of $\mathcal{G}$-equivariant extensions.
However, we recall that for a $\mathcal{G}$-invariant connection, the
equivariant characteristic classes provide canonical equivariant extensions of
the characteristic forms.

If $\pi\colon\mathcal{U}\rightarrow N$ is a principal $U(1)$ bundle and
$\Xi\in\Omega^{1}(\mathcal{U},i\mathbb{R)}$ is a connection then the curvature
form$\ \mathrm{curv}(\Xi)\in\Omega^{2}(N)$ is defined by the property
$\pi^{\ast}(\mathrm{curv}(\Xi))=\frac{i}{2\pi}d\Xi$. A a $\mathcal{G}%
$-equivariant $U(1)$-bundle is an $U(1)$-bundle $\mathcal{U}\rightarrow N$\ in
which $\mathcal{G}$ acts by $U(1)$-automorphism. A connection $\Xi\in
\Omega^{1}(\mathcal{U},i\mathbb{R)}$ on $\mathcal{U}$ is $\mathcal{G}%
$-invariant if $\phi_{\mathcal{U}}^{\ast}\Xi=\Xi$ for any $\phi\in\mathcal{G}%
$. If $\Xi$ is a $\mathcal{G}$-invariant connection then $\frac{i}{2\pi}%
D(\Xi)$ projects onto a closed $\mathcal{G}$-equivariant $2$-form
$\mathrm{curv}_{\mathcal{G}}(\Xi)\in\Omega_{\mathcal{G}}^{2}(N))$ called the
$\mathcal{G}$-equivariant curvature of $\Xi$. If $X\in\mathrm{Lie}\mathcal{G}$
then we have $\mathrm{curv}_{\mathcal{G}}(\Xi)(X)=\mathrm{curv}(\Xi)+\mu^{\Xi
}(X)$, where $\mu^{\Xi}(X)=-\frac{i}{2\pi}\Xi(X_{\mathcal{U}})$. We say that a
connection $\Xi$ is $\mathcal{G}$-flat if $\mathrm{curv}_{\mathcal{G}}(\Xi)=0$.

\section{Local forms on $\mathfrak{Met}M$\ and Jet
bundles\label{Seclocalforms}}

We recall that the classical variational calculus can be formulated in a
coordinate independent way in terms of the geometry of jet bundles (e.g. see
\cite{AndersonVB}, \cite{saunders}). A Lagrangian function is a function
$L(g)$\ that depends on the components of the metric $g_{ij}(x)$ and its
derivatives $g_{ij,I}(x)=\frac{\partial g_{ij}}{\partial x^{I}}(x)$. Hence, if
we define the jet bundle $J^{\infty}\mathcal{M}_{M}$ as the space with
coordinates $(x_{i},g_{ij},g_{ij,k},g_{ij,kr},\ldots)$, then the Lagrangian
function can be considered as a function $L\colon J^{\infty}\mathcal{M}%
_{M}\rightarrow\mathbb{R}$. More formally, we say that two metrics
$g,g^{\prime}\in\mathfrak{Met}M$ have the same same jet at $x\in M$ if in any
coordinate system we have $g_{ij}(x)=g_{ij}^{\prime}(x)$\ and $\frac{\partial
g_{ij}}{\partial x^{I}}(x)=\frac{\partial g_{ij}^{\prime}}{\partial x^{I}}(x)$
for any $i,j=1,\ldots,n$, and any symmetric multi-index $I$. For
$g\in\mathfrak{Met}M$ we denote by $j_{x}^{\infty}g$ the jet of $g$ at $x$,
that can be considered as a coordinate independent version of the Taylor
polynomial of $g$ at $x$. The space of all jets of metrics is denoted by
$J^{\infty}\mathcal{M}_{M}$. We denote the projection by $q_{\infty}\colon
J^{\infty}\mathcal{M}_{M}\rightarrow M$. A local chart $\phi\colon
U\rightarrow\mathbb{R}^{n}$ with $U\subset M$ induces a chart $\mathrm{pr}%
\phi\colon q_{\infty}^{-1}(U)\rightarrow J^{\infty}\mathcal{M}_{\mathbb{R}%
^{n}}$ by setting $\mathrm{pr}\phi(j_{x}^{\infty}g)=j_{\phi(x)}^{\infty}%
((\phi^{-1})^{\ast}g).$

A Lagrangian density is a form $\lambda\in\Omega^{n}(J^{\infty}\mathcal{M}%
_{M})$ of the type $\lambda=Ldx^{1}\wedge\ldots\wedge dx^{n}$ for $L\in
\Omega^{0}(J^{\infty}\mathcal{M}_{M})$. The action functional associated to
$\lambda$ is the function $\mathcal{A}_{\lambda}\in\Omega^{0}(\mathfrak{Met}%
M)$ defined by $\mathcal{A}_{\lambda}(g)=\int_{M}L(g)dx^{1}\wedge\ldots\wedge
dx^{n}=\int_{M}(j^{\infty}g)^{\ast}\lambda$. The variational calculus consists
in the study of the form $d\mathcal{A}^{\lambda}\in\Omega^{1}(\mathfrak{Met}%
M)\,\ $in terms of the Euler-Lagrange operator of $\lambda$. In the jet bundle
approach, $d\mathcal{A}^{\lambda}$\ is determined by the Euler-Lagrange form
$\mathcal{E}(\lambda)\in\Omega^{n+1}(J^{\infty}\mathcal{M}_{M})$ given in
local coordinates by $\mathcal{E}(\lambda)=\left(  \sum_{I}(-1)^{|I|}\frac
{d}{dx^{I}}\left(  \frac{\partial L}{\partial g_{ij,I}}\right)  \right)
dg_{ij}\wedge dx^{1}\wedge\ldots\wedge dx^{n}$, where $\frac{d}{dx^{k}}%
=\frac{\partial}{\partial x^{k}}+\sum_{I}g_{ij,I+k}\frac{\partial}{\partial
g_{ij,I}}$. As $\mathfrak{Met}M$ is an open convex set of the space
$\mathcal{S}^{2}M$ of bilinear symmetric\ tensor on $M$ we have $T_{g}%
\mathfrak{Met}M\simeq\mathcal{S}^{2}M$ for any $g\in\mathfrak{Met}M$. Then the
form $d\mathcal{A}^{\lambda}$ is given by $d\mathcal{A}_{g}^{\lambda}%
(h)=\int_{M}\left(  \sum_{I}(-1)^{|I|}\frac{d}{dx^{I}}(\frac{\partial
L}{\partial g_{ij,I}}\right)  h_{ij}dx^{1}\wedge\ldots\wedge dx^{n}=\int
_{M}(j^{\infty}g)^{\ast}(\iota_{H}\mathcal{E}(\lambda))$, where $H$ is the
vector field $H=h_{ij}\frac{\partial}{\partial g_{ij}}+\sum_{I}\frac{\partial
h_{ij}}{\partial x^{I}}\frac{\partial}{\partial g_{ij,I}}\in\mathfrak{X}%
(J^{\infty}\mathcal{M}_{M})$. In particular $d\mathcal{A}_{g}^{\lambda}=0$ if
and only if $\mathcal{E}(\lambda)(j_{x}^{\infty}g)=0$ for all $x\in M$, and
the last are just the usual Euler-Lagrange equations.

The correspondence $\lambda\in\Omega^{n}(J^{\infty}\mathcal{M}_{M}%
)\mapsto\mathcal{A}^{\lambda}\in\Omega^{0}(\mathfrak{Met}M)$, and
$\mathcal{E}(\lambda)\in\Omega^{n+1}(J^{\infty}\mathcal{M}_{M})\mapsto
d\mathcal{A}^{\lambda}\in\Omega^{1}(\mathfrak{Met}M)\ $is extended to forms of
arbitrary degree in \cite{equiconn}. Let $\mathrm{j}^{\infty}\colon
M\times\mathfrak{Met}M\rightarrow J^{\infty}\mathcal{M}_{M}$, $\mathrm{j}%
^{\infty}(x,g)=j_{x}^{\infty}g$ be the evaluation map. We define $\Im
\colon\Omega^{n+k}(J^{\infty}\mathcal{M}_{M})\longrightarrow\Omega
^{k}(\mathfrak{Met}M)$, by $\Im(\alpha)=\int_{M}\left(  \mathrm{j}^{\infty
}\right)  ^{\ast}\alpha$ for $\alpha\in\Omega^{n+k}(J^{\infty}\mathcal{M}%
_{M})$. If $\alpha\in\Omega^{n}(J^{\infty}\mathcal{M}_{M})$ we have
$\Im(\alpha)_{g}=\int_{M}\left(  j^{\infty}g\right)  ^{\ast}\alpha$, and if
$\alpha\in\Omega^{n+k}(J^{\infty}\mathcal{M}_{M})$ we have $\Im(\alpha
)_{g}(h_{1},\ldots,h_{k})=\int_{M}\left(  j^{\infty}g\right)  ^{\ast}%
(\iota_{H_{k}}\cdots\iota_{H_{1}}\alpha)$ for $h_{1},\ldots,h_{k}%
\in\mathcal{S}^{2}M$. We define the space of local $k$-forms on
$\mathfrak{Met}M$ by $\Omega_{\mathrm{loc}}^{k}(\mathfrak{Met}M)=\Im
(\Omega^{n+k}(J^{\infty}\mathcal{M}_{M}))\subset\Omega^{k}(\mathfrak{Met}M)$.
If $\alpha\in\Omega^{k}(J^{\infty}\mathcal{M}_{M})$ with $k<n$ we define
$\Im(\alpha)=0$.\ We have $\Im(d\alpha)=d\Im(\alpha)$, and hence if $\alpha
\in\Omega_{\mathrm{loc}}^{k}(\mathfrak{Met}M)$ then $d\alpha\in\Omega
_{\mathrm{loc}}^{k+1}(\mathfrak{Met}M)$. The local cohomology of
$\mathfrak{Met}M$, $H_{\mathrm{loc}}^{\bullet}(\mathfrak{Met}M)$, is the
cohomology of $(\Omega_{\mathrm{loc}}^{\bullet}(\mathfrak{Met}M),d)$.

The integration operator extends to a map on equivariant differential forms
(see \cite{equiconn}) $\Im\colon\Omega_{\mathcal{D}_{M}^{+}}^{n+k}(J^{\infty
}\mathcal{M}_{M})\rightarrow\Omega_{\mathcal{D}_{M}^{+}}^{k}(\mathfrak{Met}%
M)$, by setting $(\Im(\alpha))(X)=\Im(\alpha(X))$ for any $\alpha\in
\Omega_{\mathcal{D}_{M}^{+}}^{n+k}(J^{\infty}\mathcal{M}_{M})$ and
$X\in\mathfrak{X}(M)$. The map $\Im$ induces a homomorphism in equivariant
cohomology $\Im\colon H_{\mathcal{D}_{M}^{+}}^{n+k}(J^{\infty}\mathcal{M}%
_{M})\rightarrow H_{\mathcal{D}_{M}^{+}}^{k}(\mathfrak{Met}M)$. We define the
space of local $\mathcal{D}_{M}^{+}$-equivariant $q$-forms on $\mathfrak{Met}%
M$ by $\Omega_{\mathcal{D}_{M}^{+},\mathrm{loc}}^{q}(\mathfrak{Met}%
M)=\bigoplus_{2k+r=q}(\mathcal{P}^{k}(\mathfrak{X}(M),\Omega_{\mathrm{loc}%
}^{r}(\mathfrak{Met}M)))^{\mathcal{D}_{M}^{+}}\subset\Omega_{\mathcal{D}%
_{M}^{+}}^{q}(\mathfrak{Met}M)$, and the cohomology $H_{\mathcal{D}_{M}%
^{+},\mathrm{loc}}^{\bullet}(\mathfrak{Met}M)$\ of $(\Omega_{\mathcal{D}%
_{M}^{+},\mathrm{loc}}^{\bullet}(\mathfrak{Met}M),D)$ is called the local
$\mathcal{D}_{M}^{+}$-equivariant cohomology of $\mathfrak{Met}M$. Then $\Im$
induces a cochain map $\Im\colon\Omega_{\mathcal{D}_{M}^{+}}^{n+k}(J^{\infty
}\mathcal{M}_{M})\rightarrow\Omega_{\mathcal{D}_{M}^{+},\mathrm{loc}}%
^{k}(\mathfrak{Met}M)$.

\subsection{Local cohomology and the variational bicomplex}

On $J^{\infty}\mathcal{M}_{M}$ we have a natural bigrading $\Omega
^{k}(J^{\infty}\mathcal{M}_{M})=\bigoplus_{k=p+q}\Omega^{p,q}(J^{\infty
}\mathcal{M}_{M})$ into horizontal and contact (or vertical) degree. This
bigrading is the infinitesimal version of the bigrading corresponding to the
product structure on $M\times\mathfrak{Met}M$, and the map $\mathrm{j}%
^{\infty}\colon M\times\mathfrak{Met}M\rightarrow J^{\infty}\mathcal{M}_{M}$
preserves bidegree. If $\alpha\in\Omega^{k}(J^{\infty}\mathcal{M}_{M})$ we
denote by $\alpha_{p,q}\in\Omega^{p,q}(J^{\infty}\mathcal{M}_{M})$ its
$p$-horizontal and $q$-contact component. According to the preceding bigrading
we have a decomposition of the exterior differential $d=d_{H}+d_{V}$.

For $k>0$ we denote by $I\colon\Omega^{n,k}(J^{\infty}\mathcal{M}%
_{M})\rightarrow\Omega^{n,k}(J^{\infty}\mathcal{M}_{M})$ the interior Euler
operator. We recall (e.g. see \cite{AndersonVB}) that it satisfies the
following properties: $I^{2}=I$, $\mathrm{ker}I=d_{H}(\Omega^{n-1,k})$,
$Id_{V}=d_{V}I$. The image of the interior Euler operator $\mathcal{F}%
^{k}(J^{\infty}\mathcal{M}_{M})=I(\Omega^{n,k}(J^{\infty}\mathcal{M}_{M}))$ is
called the space of functional $k$-forms. The vertical differential $d_{V}$
induces a differential on the space of functional forms $\delta_{V}%
\colon\mathcal{F}^{k}(J^{\infty}\mathcal{M}_{M})\rightarrow\mathcal{F}%
^{k+1}(J^{\infty}\mathcal{M}_{M})$, $\delta_{V}\alpha=I(d_{V}\alpha)$. In this
context a lagrangian density is a form $\lambda\in\Omega^{n,0}(J^{\infty
}\mathcal{M}_{M})$ and $\mathcal{E}(\lambda)=I(d_{V}\lambda)\in\mathcal{F}%
^{1}(J^{\infty}\mathcal{M}_{M})$ is the Euler-Lagrange operator of $\lambda$.

The relationship between the variational bicomplex and the local forms is the
following (see \cite{localVB} ). If $\alpha\in\Omega^{n+k}(J^{\infty
}\mathcal{M}_{M})$, $k>0$, we have $\Im(\alpha)=\Im(\alpha_{n,k})=\Im
(I(\alpha_{n,k}))$. Furthermore, the restriction of $\Im$ to $\mathcal{F}%
^{k}(J^{\infty}\mathcal{M}_{M})$ induces isomorphisms $\mathcal{F}%
^{k}(J^{\infty}\mathcal{M}_{M})\simeq\Omega_{\mathrm{loc}}^{k}(\mathfrak{Met}%
M)$ for $k>0$. Hence, the local forms of degree $k>0$ can be studied\ in terms
of its\ canonical representative on the jet bundle. Unfortunately, for $k=0$
we do not have a similar operator, i.e. for a local functional $\Lambda
\in\Omega_{\mathrm{loc}}^{0}(\mathfrak{Met}M)$ we cannot select a canonical
Lagrangian density $\lambda$ such that $\Im(\lambda)=\Lambda$.

The isomorphism $\mathcal{F}^{k}(J^{\infty}\mathcal{M}_{M})\simeq
\Omega_{\mathrm{loc}}^{k}(\mathfrak{Met}M)$, combined with the classical
calculation of the cohomology of the variational bicomplex implies that the
map $\Im$ induces isomorphisms $H_{\mathrm{loc}}^{k}(\mathfrak{Met}M)\simeq
H^{k}(\mathcal{F}^{\bullet}(J^{\infty}\mathcal{M}_{M}))\simeq H^{n+k}%
(J^{\infty}\mathcal{M}_{M})=0$ for $k>0$ (see \cite{localVB} for details).

For any $\phi\in\mathcal{D}_{M}$ we define $\mathrm{pr}\phi\in\mathcal{D}%
_{J^{\infty}\mathcal{M}_{M}}\ $by setting $\mathrm{pr}\phi(j_{x}^{\infty
}g)=j_{\phi(x)}^{\infty}(\phi\cdot g)$ for any $g\in\mathfrak{Met}M$. The
action of $\mathcal{D}_{M}^{+}$ commutes with $\Im$ and hence we have
isomorphisms $\Omega_{\mathrm{loc}}^{k}(\mathfrak{Met}M)^{\mathcal{D}_{M}^{+}%
}\simeq\mathcal{F}^{k}(J^{\infty}\mathcal{M}_{M})^{\mathcal{D}_{M}^{+}}$ for
$k>0$ and $H_{\mathrm{loc}}^{k}(\mathfrak{Met}M)^{\mathcal{D}_{M}^{+}}\simeq
H^{k}(\mathcal{F}^{\bullet}(J^{\infty}\mathcal{M}_{M})^{\mathcal{D}_{M}^{+}})$
for $k>1$.

The space $\Omega^{k}(\mathfrak{Met}M)^{\mathcal{D}_{M}^{+}}$ can depend on
the global properties of the manifold $M$. A big difference of local
cohomology and ordinary cohomology is that the space of local forms
$\Omega_{\mathrm{loc}}^{k}(\mathfrak{Met}M)^{\mathcal{D}_{M}^{+}}$ only
depends on the dimension of $M$.

\begin{lemma}
\label{LemmaUniver}Let $\alpha\in\Omega^{q}(J^{\infty}\mathcal{M}%
_{\mathbb{R}^{n}})^{\mathcal{D}_{\mathbb{R}^{n}}^{+}}$. For any $j_{x}%
^{\infty}g\in J^{\infty}\mathcal{M}_{M}$ we choose an oriented local chart
$\phi\colon U\rightarrow\mathbb{R}^{n}$ with $x\in U$ and the induced chart
$\mathrm{pr}\phi\colon q_{\infty}^{-1}(U)\rightarrow J^{\infty}\mathcal{M}%
_{\mathbb{R}^{n}}$. The form $\alpha_{M}(j_{x}^{\infty}g)\in\Omega
_{j_{x}^{\infty}g}^{q}(J^{\infty}\mathcal{M}_{M})$ defined by $\alpha
_{M}=(\mathrm{pr}\phi)^{\ast}\alpha$ is independent of $\phi$ and $\alpha_{M}$
is $\mathcal{D}_{M}^{+}$-invariant . In this way we obtain a map that induces
isomorphisms\
\[
\Omega^{q}(J^{\infty}\mathcal{M}_{\mathbb{R}^{n}})^{\mathcal{D}_{\mathbb{R}%
^{n}}^{+}}\simeq\Omega^{q}(J^{\infty}\mathcal{M}_{M})^{\mathcal{D}_{M}^{+}%
}\simeq\Omega^{q}(J^{\infty}\mathcal{M}_{M})^{\mathcal{D}_{M}^{0}}.
\]

\end{lemma}

\begin{proof}
That $\alpha_{M}$ is well defined and it is $\mathcal{D}_{M}^{+}$-invariant
easily follows from the $\mathcal{D}_{\mathbb{R}^{n}}^{+}$-invariance of
$\alpha$. The maps $\Omega^{q}(J^{\infty}\mathcal{M}_{\mathbb{R}^{n}%
})^{\mathcal{D}_{\mathbb{R}^{n}}^{+}}\rightarrow\Omega^{q}(J^{\infty
}\mathcal{M}_{M})^{\mathcal{D}_{M}^{+}}$ and $\Omega^{q}(J^{\infty}%
\mathcal{M}_{\mathbb{R}^{n}})^{\mathcal{D}_{\mathbb{R}^{n}}^{+}}%
\rightarrow\Omega^{q}(J^{\infty}\mathcal{M}_{M})^{\mathcal{D}_{M}^{+}}$ are
clearly injective. That they are surjective follows from the fact that any
germ of orientation preserving\ diffeomorphism can be extended to a
diffeomorphism isotopic to the identity (e.g. see \cite{Palais}).
\end{proof}

If we apply the preceding result to the functional forms we obtain the following

\begin{corollary}
\label{Corollarylocal}We have the isomorphisms
\begin{equation}
\Omega_{\mathrm{loc}}^{k}(\mathfrak{Met}M)^{\mathcal{D}_{M}^{0}}\simeq
\Omega_{\mathrm{loc}}^{k}(\mathfrak{Met}M)^{\mathcal{D}_{M}^{+}}%
\simeq\mathfrak{F}^{k}(J^{\infty}\mathcal{M}_{\mathbb{R}^{n}})^{\mathcal{D}%
_{\mathbb{R}^{n}}^{+}} \label{Iso}%
\end{equation}
for $k>0$ and
\begin{equation}
H_{\mathrm{loc}}^{k}(\mathfrak{Met}M)^{\mathcal{D}_{M}^{0}}\simeq
H_{\mathrm{loc}}^{k}(\mathfrak{Met}M)^{\mathcal{D}_{M}^{+}}\simeq
H^{k}(\mathfrak{F}^{\bullet}(J^{\infty}\mathcal{M}_{\mathbb{R}^{n}%
})^{\mathcal{D}_{\mathbb{R}^{n}}^{+}}) \label{Iso2}%
\end{equation}
for $k>1.$
\end{corollary}

Hence, we can consider $\mathfrak{F}^{k}(J^{\infty}\mathcal{M}_{\mathbb{R}%
^{n}})^{\mathcal{D}_{\mathbb{R}^{n}}^{+}}$ as the space of universal local
$k\,$-forms on dimension $n$. We remark that this result is not valid for
$H_{\mathrm{loc}}^{1}(\mathfrak{Met}M)^{\mathcal{D}_{M}^{+}}$ (see Section
\ref{SectionVB-CS}).

\begin{remark}
If $\mathcal{G}$ is a subgroup such that $\mathcal{D}_{M}^{0}\subset
\mathcal{G}\subset\mathcal{D}_{M}^{+}$ we also have that $\Omega^{q}%
(J^{\infty}\mathcal{M}_{M})^{\mathcal{G}}\simeq\Omega^{q}(J^{\infty
}\mathcal{M}_{\mathbb{R}^{n}})^{\mathcal{D}_{\mathbb{R}^{n}}^{+}}$ and
$\Omega_{\mathrm{loc}}^{k}(\mathfrak{Met}M)^{\mathcal{G}}\simeq\mathfrak{F}%
^{k}(J^{\infty}\mathcal{M}_{\mathbb{R}^{n}})^{\mathcal{D}_{\mathbb{R}^{n}}%
^{+}}$. This applies for example to the subgroup of elements of $\mathcal{D}%
_{M}^{+}$ preserving a spin structure on $M$.
\end{remark}

\subsection{Pontryagin forms on the jet of metrics}

The pull-back bundle $\bar{q}_{\infty}\colon q_{\infty}^{\ast}FM\rightarrow
J^{\infty}\mathcal{M}_{M}$ is a $\mathcal{D}_{M}$-equivariant principal
$Gl(n,\mathbb{R})$-bundle. As the Levi-Civita connection $\omega^{g}$\ of a
metric $g\in\mathfrak{Met}M$\ depends only on the first derivatives of $g$, we
can define in a natural way a connection on $q_{\infty}^{\ast}FM$ by setting $%
\mbox{\boldmath$\omega$}%
(X)=\omega^{g}((\bar{q}_{\infty})_{\ast}(X))$ for $X\in T_{(j_{x}^{\infty
}g,u)}(q_{\infty}^{\ast}FM)$, $g\in\mathfrak{Met}M$ and $u\in(FM)_{x}$. It can
be seen (see \cite{natconn}) that $%
\mbox{\boldmath$\omega$}%
$ is $\mathcal{D}_{M}$-invariant. If $p\in I^{k}(O(n))\subset I^{k}%
(Gl(n,\mathbb{R}))$ and $%
\mbox{\boldmath$\Omega$}%
$ is the curvature form of $%
\mbox{\boldmath$\omega$}%
$, we define the form $p(%
\mbox{\boldmath$\Omega$}%
)\in\Omega^{2k}(J^{\infty}\mathcal{M}_{M})^{\mathcal{D}_{M}}$. In particular,
for the Pontryagin polynomial $p_{k}$ we have the Pontryagin form $p_{k}(%
\mbox{\boldmath$\Omega$}%
)\in\Omega^{4k}(J^{\infty}\mathcal{M}_{M})^{\mathcal{D}_{M}}$.

\begin{remark}
\emph{We can define the Pontryagin forms by using the connection} $%
\mbox{\boldmath$\omega$}%
$ \emph{because we consider }$I^{k}(O(n))\subset I^{k}(Gl(n,R))$\emph{. But in
this way we cannot define the Euler form. It is shown in \cite{natconn} that
the connection }$%
\mbox{\boldmath$\omega$}%
$\emph{\ can be modified to obtain a} $\mathcal{D}_{M}$\emph{-invariant
connection }$%
\mbox{\boldmath$\omega$}%
^{\prime}$\emph{\ on the principal }$O(n)$\emph{-bundle }$OM\rightarrow
J^{\infty}\mathcal{M}_{M}$ \emph{where }$OM=\left\{  (j_{x}^{\infty}%
g,u_{x})\in q_{\infty}^{\ast}FM\colon\text{ }u_{x}\,\text{\emph{is} }%
g_{x}\text{\emph{-orthonormal}}\right\}  $. \emph{By using the connection }$%
\mbox{\boldmath$\omega$}%
^{\prime}$\emph{ we can define the Pontryagin forms and the Euler form }%
$E\in\Omega^{n}(J^{\infty}\mathcal{M}_{M})^{\mathcal{D}_{M}}$\emph{. In our
study of anomaly cancellation we only\ consider forms of degree }$n+2$\emph{
and }$n+1$\emph{, and hence the Euler form does not appear. Furthermore, it
can be seen that the Pontryagin forms defined by using the connections }$%
\mbox{\boldmath$\omega$}%
$\emph{ and }$%
\mbox{\boldmath$\omega$}%
^{\prime}$\emph{ lead to equivalent results. Hence we use the connection }$%
\mbox{\boldmath$\omega$}%
$\emph{ that it is easier to define.}
\end{remark}

By applying the map $\Im$ to the forms $p(%
\mbox{\boldmath$\Omega$}%
)$ we obtain closed local forms $\sigma^{p}=\Im(p(%
\mbox{\boldmath$\Omega$}%
))\in\Omega_{\mathrm{loc}}^{2k-n}(\mathfrak{Met}M)^{\mathcal{D}_{M}^{+}}$ for
$2k\geq n$. We can consider another interpretation of the forms $\sigma^{p}%
$\ that is more familiar in the study of gravitational anomalies (e.g. see
\cite{AS}, \cite{Kel}). The evaluation map determines a map $\overline
{\mathrm{ev}}\colon$ $FM\times\mathfrak{Met}M\rightarrow q_{\infty}^{\ast}FM$.
The connection $%
\mbox{\boldmath$\omega$}%
^{\ast}=\overline{\mathrm{ev}}^{\ast}%
\mbox{\boldmath$\omega$}%
$ defines a connection on this bundle and we have $\sigma^{p}=\int_{M}p(%
\mbox{\boldmath$\Omega$}%
^{\ast})$. It is shown in \cite{localVB} that we have the following

\begin{proposition}
\label{injective}The map $I^{k}(O(n))\rightarrow H_{\mathrm{loc}}%
^{2k-n}(\mathfrak{Met}M)^{\mathcal{D}_{M}^{+}}$, $p\mapsto\lbrack\sigma^{p}]$
is injective for $n+1<2k\leq2n$.
\end{proposition}

Let $A_{0}$ be a connection on $FM\rightarrow M$. As $%
\mbox{\boldmath$\omega$}%
$ and $\overline{A}_{0}=q_{\infty}^{\ast}A_{0}$ are connections on the same
bundle $q_{\infty}^{\ast}FM\rightarrow J^{\infty}\mathcal{M}_{M}$ we have $p(%
\mbox{\boldmath$\Omega$}%
)-q_{\infty}^{\ast}p(F_{0})=dTp(%
\mbox{\boldmath$\omega$}%
,\overline{A}_{0})$. But if $2k>n$ then $p(F_{0})=0$ by dimensional reasons
and hence we have $p(%
\mbox{\boldmath$\Omega$}%
)=dTp(%
\mbox{\boldmath$\omega$}%
,\overline{A}_{0})$ and $\sigma^{p}=d\rho^{p}$,\ where $\rho^{p}=\Im(Tp(%
\mbox{\boldmath$\omega$}%
,\overline{A}_{0}))\in\Omega_{\mathrm{loc}}^{2k-n-1}(\mathfrak{Met}M)$. Note
that the form $\rho^{p}$ is not $\mathcal{D}_{M}^{+}$-invariant as it depends
on the connection $A_{0}$. In particular, if $n=3\operatorname{mod}4$ and
$p\in I^{(n+1)/2}(O(n))$ then the function $\rho^{p}=\Im(Tp(%
\mbox{\boldmath$\omega$}%
,\overline{A}_{0}))\in\Omega_{\mathrm{loc}}^{0}(\mathfrak{Met}M)$ is given by
the Chern-Simons action
\[
\rho^{p}(g)=\int_{M}(j^{\infty}g)^{\ast}Tp(%
\mbox{\boldmath$\omega$}%
,\overline{A}_{0})=\int_{M}Tp(\omega^{g},A_{0})=CS_{p.A_{0}}(g)\text{.}%
\]
Furthermore we have $dCS_{p.A_{0}}=d\rho^{p}=\sigma^{p}\in\Omega
_{\mathrm{loc}}^{1}(\mathfrak{Met}M)^{\mathcal{D}_{M}^{+}}$, as commented in
Section \ref{SectionCS}. Moreover $\lambda_{p,A_{0}}=Tp(%
\mbox{\boldmath$\omega$}%
,\overline{A}_{0})_{n,0}\in\Omega^{n,0}(J^{\infty}\mathcal{M}_{M})$ is the
Lagrangian density corresponding to the Chern-Simons action.

\begin{remark}
\emph{The forms }$p(%
\mbox{\boldmath$\Omega$}%
)$\emph{ are invariant under arbitrary diffeomorphisms of }$M$\emph{. However,
the integration map }$\Im$\emph{ is defined by using an orientation on }%
$M$\emph{. Hence the forms }$\sigma^{p}$\emph{ are not invariant under
orientation reversing diffeomorphisms. As in Section \ref{SectionCS}\ we can
solve this problem by considering the space} $\widetilde{\mathfrak{Met}}M$.
\emph{The integration map defines a map} $\widetilde{\Im}\colon\Omega
^{n+k}(J^{\infty}\mathcal{M}_{M})\rightarrow\Omega^{k}(\widetilde
{\mathfrak{Met}}M)$ \emph{and this map is} $\mathcal{D}_{M}$%
\emph{-equivariant. We define }$\Omega_{\mathrm{loc}}^{k}(\widetilde
{\mathfrak{Met}}M)=\widetilde{\Im}(\Omega^{n+k}(J^{\infty}\mathcal{M}_{M}))$
\emph{and we have} $\widetilde{\sigma^{p}}=\widetilde{\Im}(p(%
\mbox{\boldmath$\Omega$}%
))\in\Omega_{\mathrm{loc}}^{k}(\widetilde{\mathfrak{Met}}M)^{\mathcal{D}_{M}}$.
\end{remark}

\subsection{Chern-Simons actions and generalized Cotton
tensors\label{SectionVB-CS}}

Let $M$ be an oriented compact manifold of dimension $n=3\operatorname{mod}4$
and $p\in I^{(n+1)/2}(O(n))$. Under the isomorphism $\Omega_{\mathrm{loc}}%
^{1}(\mathfrak{Met}M)^{\mathcal{D}_{M}^{+}}\simeq\mathcal{F}^{1}(J^{\infty
}\mathcal{M}_{M})^{\mathcal{D}_{M}^{+}}$ the form $\sigma^{p}$ corresponds to
the functional $1$-form $\mathcal{E}^{p}=I(p(%
\mbox{\boldmath$\Omega$}%
)_{n,1})\in\mathfrak{F}^{1}(J^{\infty}\mathcal{M}_{M})^{\mathcal{D}_{M}^{+}}$,
which is called the generalized Cotton Tensor (e.g. see \cite{MetricsAnderson}%
). We remark that $\mathcal{E}^{p}$ can also be obtained as the Euler-Lagrange
operator of the Chern-Simons Lagrangian $\lambda_{p,A_{0}}$. It is known that
$\mathcal{E}^{p}$ is not the Euler-Lagrange operator of a invariant
$\mathcal{D}_{M}^{+}$ Lagrangian density (see \cite{MetricsAnderson},
\cite{localVB}). Furthermore, we have the following result (see
\cite{MetricsAnderson})

\begin{proposition}
\label{anderson}a) If $n\neq3\operatorname{mod}4$ any $\delta_{V}$-closed
$\mathcal{D}_{\mathbb{R}^{n}}^{+}$-invariant functional $1$-form
$\mathcal{E}\in\mathfrak{F}^{1}(J^{\infty}\mathcal{M}_{\mathbb{R}^{n}%
})^{\mathcal{D}_{\mathbb{R}^{n}}^{+}}$\ is the Euler-Lagrange operator of a
$\mathcal{D}_{\mathbb{R}^{n}}^{+}$-invariant lagrangian density.

b) If $n=3\operatorname{mod}4$ then any $\delta_{V}$-closed $\mathcal{D}%
_{\mathbb{R}^{n}}^{+}$-invariant functional $1$-form $\mathcal{E}%
\in\mathfrak{F}^{1}(J^{\infty}\mathcal{M}_{\mathbb{R}^{n}})^{\mathcal{D}%
_{\mathbb{R}^{n}}^{+}}$\ is of the form $\mathcal{E}=\mathcal{E}%
^{p}+\mathcal{E}(\lambda)$ for $p\in I^{(n+1)/2}(O(n))$ and $\lambda\in
\Omega^{n,0}(J^{\infty}\mathcal{M}_{\mathbb{R}^{n}})^{\mathcal{D}%
_{\mathbb{R}^{n}}^{+}}\ $a $\mathcal{D}_{\mathbb{R}^{n}}^{+}$-invariant
lagrangian density.

c)The results of a) and b) are also valid if $\mathcal{D}_{\mathbb{R}^{n}}%
^{+}$ is replaced with $\mathcal{D}_{\mathbb{R}^{n}}$.
\end{proposition}

By Lemma \ref{LemmaUniver} this result also holds for any compact manifold $M$
of dimension $n$. By using Corollary \ref{Corollarylocal} we obtain the following

\begin{theorem}
\label{H1}Let $M$ be a compact oriented manifold of dimension $n$.

a) If $n\neq3\operatorname{mod}4$ then we have $H_{\mathrm{loc}}%
^{1}(\mathfrak{Met}M)^{\mathcal{D}_{M}^{+}}=0$.

b) If $n\!=\!3\operatorname{mod}4$ then any closed $\mathcal{D}_{M}^{+}%
$-invariant local $1$-form $\sigma\!\in\!\Omega_{\mathrm{loc}}^{1}%
(\mathfrak{Met}M)^{\mathcal{D}_{M}^{+}}$\ is of the form $\sigma=\sigma
^{p}+d\rho\mathcal{\ }$for $p\in I^{(n+1)/2}(O(n))$ and $\rho\in
\Omega_{\mathrm{loc}}^{0}(\mathfrak{Met}M)^{\mathcal{D}_{M}^{+}}$.
\end{theorem}

We conclude that $H_{\mathrm{loc}}^{1}(\mathfrak{Met}M)^{\mathcal{D}_{M}^{+}}$
is generated by the Pontryagin forms, as commented in the Introduction. The
Corollary \ref{CSD0} can be reobtained by using the properties of the forms
$\sigma^{p}$. We have the following

\begin{lemma}
\label{basic}If $p\in I^{(n+1)/2}(O(n))$ then$\ \sigma^{p}=\Im(p(%
\mbox{\boldmath$\Omega$}%
))\in\Omega_{\mathrm{loc}}^{1}(\mathfrak{Met}M)^{\mathcal{D}_{M}^{+}}$\ is
$\mathcal{D}_{M}^{+}$-basic, i.e. $d\sigma^{p}=0$ and $\iota
_{X_{\mathfrak{Met}M}}\sigma^{p}=0$\ for any $X\in\mathfrak{X}(M)$.
\end{lemma}

\begin{proof}
The result follows from the existence of equivariant Pontryagin forms. As the
connection $%
\mbox{\boldmath$\omega$}%
$ is $\mathcal{D}_{M}^{+}$-invariant, the form $p(%
\mbox{\boldmath$\omega$}%
)$ admits a canonical equivariant extension $\varpi^{p}$ (the explicit
expression of this extension can be found in \cite{WP}). The $D$-closed
equivariant $1$-form $\Im(\varpi^{p})\in\Omega_{\mathcal{D}_{M}^{+}%
,\mathrm{loc}}^{1}(\mathfrak{Met}M)\simeq\Omega_{\mathrm{loc}}^{1}%
(\mathfrak{Met}M)^{\mathcal{D}_{M}^{+}}$ is given simply by the $1$-form
$\sigma^{p}=\Im(p(%
\mbox{\boldmath$\Omega$}%
))$. Hence we have $d\sigma^{p}=0$ and $\iota_{X_{\mathfrak{Met}M}}\sigma
^{p}=0$ for $X\in\mathfrak{X}(M)$.
\end{proof}

This result implies Corollary \ref{CSD0}, as we have $L_{X_{\mathfrak{Met}M}%
}CS_{p,A_{0}}=\iota_{X_{\mathfrak{Met}M}}dCS_{p,A_{0}}=\iota
_{X_{\mathfrak{Met}M}}\sigma^{p}=0$ for $X\in\mathfrak{X}(M)$.

By using Corollary \ref{Corollarylocal}\ and Theorem \ref{H1} we obtain the following

\begin{corollary}
For any compact oriented manifold we have $H_{\mathrm{loc}}^{1}(\mathfrak{Met}%
M)^{\mathcal{D}_{M}^{0}}=0$.
\end{corollary}

Hence we could have $H_{\mathrm{loc}}^{1}(\mathfrak{Met}M)^{\mathcal{D}%
_{M}^{0}}\neq H_{\mathrm{loc}}^{1}(\mathfrak{Met}M)^{\mathcal{D}_{M}^{+}}$ in
dimension $n\neq3\operatorname{mod}4$.

\section{Locality and anomaly cancellation\label{SectionLocAnomalyCanc}}

As commented in the introduction, in \cite{AnomaliesG} the necessary and
sufficient condition for anomaly cancellation can be expressed in terms of the
equivariant curvature and holonomy of the Bismut-Freed connection. Precisiely,
let $\pi\colon\mathcal{U}\rightarrow\mathfrak{Met}M$ be a $\mathcal{D}_{M}%
^{+}$-equivariant $U(1)$-bundle with a $\mathcal{D}_{M}^{+}$-invariant
connection $\Xi$. If $\phi\in\mathcal{D}_{M}^{+}$ we define $\mathcal{C}%
^{\phi}=\{\gamma\colon\lbrack0,1]\rightarrow\mathfrak{Met}M:\gamma
(1)=\phi\cdot\gamma(0)\}$. For any $\phi\in\mathcal{D}_{M}^{+}$ and $\gamma
\in\mathcal{C}^{\phi}$ we define the $\phi$-equivariant holonomy of $\Xi$
$\mathrm{hol}_{\phi}^{\Xi}(\gamma)\in\mathbb{R}/\mathbb{Z}$ by the condition
$\overline{\gamma}(1)=\phi_{\mathcal{U}}(\overline{\gamma}(0))\exp(2\pi
i\cdot\mathrm{hol}_{\phi}^{\Xi}(\gamma))$, where $\overline{\gamma}$ is any
$\Xi$-horizontal lift of $\gamma$. Then we have the following (see
\cite{AnomaliesG})

\begin{proposition}
\label{TopAno} Let $\mathcal{U}\rightarrow\mathfrak{Met}M$ be a $\mathcal{D}%
_{M}^{+}$-equivariant $U(1)$-bundle with a $\mathcal{D}_{M}^{+}$-invariant
connection $\Xi$.

p) $\mathcal{U}\rightarrow\mathfrak{Met}M$ admits a $\mathcal{D}_{M}^{0}%
$-equivariant section if and only if there exists a $\mathcal{D}_{M}^{0}%
$-invariant $1$-form $\beta\in\Omega^{1}(\mathfrak{Met}M)^{\mathcal{D}_{M}%
^{0}}$ such that $\mathrm{curv}_{\mathcal{D}_{M}^{0}}(\Xi)=D\beta$.

g) $\mathcal{U}\rightarrow\mathfrak{Met}M$ admits a $\mathcal{D}_{M}^{+}%
$-equivariant section if and only if there exists a $\mathcal{D}_{M}^{+}%
$-invariant $1$-form $\beta\in\Omega^{1}(\mathfrak{Met}M)^{\mathcal{D}_{M}%
^{+}}$ such that $\mathrm{hol}_{\phi}^{\Xi}(\gamma)=\int_{\gamma}%
\beta\operatorname{mod}\mathbb{Z}$ for any $\gamma\in\mathcal{C}^{\phi}$.
\end{proposition}

When the preceding result is applied to the anomaly bundle, it provides
necesary and sufficient conditions for topological anomaly cancellation, i.e,
to the posibility of solving condition (\ref{Lambda}) with a term $\Lambda
\in\Omega^{0}(\mathfrak{Met}M)$. However, physical anomaly cancellation
requires that $\Lambda$ should be a local functional $\Omega_{\mathrm{loc}%
}^{0}(\mathfrak{Met}M)$. To generalize the Proposition \ref{TopAno} for
physical anomalies we need analogous local\ versions of all the objects
appearing on it. This has been done in \cite{anomalies} for perturbative
anomalies and in \cite{AnomaliesG} for global anomalies. We say that a
$\mathcal{D}_{M}^{+}$-invariant connection $\Xi$ on $\mathcal{U}%
\rightarrow\mathfrak{Met}M$ is local if $\mathrm{curv}_{\mathcal{D}_{M}^{+}%
}(\Xi)\in\Omega_{\mathcal{D}_{M}^{+},\mathrm{loc}}^{2}(\mathfrak{Met}M)$. Note
that when $\Xi$ is the Bismut-Freed connection on a determinant bundle, then
$\Xi$ is a local connection by the equivariant\ Atiyah-Singer theorem for
families (see \cite{FreedEqui}). Finally it is needed to determine
intrinsically what kind of sections of $\mathcal{U}\rightarrow\mathfrak{Met}M$
correspond to the solutions of equation (\ref{Lambda}) with $\Lambda$ a local
functional. It is shown in \cite{AnomaliesG} that this condition is that a
section $S\colon\mathfrak{Met}M\rightarrow\mathcal{U}$ is $\Xi$-local if
$S^{\ast}(\Xi)\in\Omega_{\mathrm{loc}}^{1}(\mathfrak{Met}M)$. With this
definitions we have the following result

\begin{proposition}
\label{PysAno} Let $\mathcal{U}\rightarrow\mathfrak{Met}M$ be a $\mathcal{D}%
_{M}^{+}$-equivariant $U(1)$-bundle with a local $\mathcal{D}_{M}^{+}%
$-invariant connection $\Xi$.

P) $\mathcal{U}\rightarrow\mathfrak{Met}M$ admits a $\mathcal{D}_{M}^{0}%
$-equivariant $\Xi$-local section if and only if there exists a $\mathcal{D}%
_{M}^{0}$-invariant local $1$-form $\beta\in\Omega_{\mathrm{loc}}%
^{1}(\mathfrak{Met}M)^{\mathcal{D}_{M}^{0}}$ such that $\mathrm{curv}%
_{\mathcal{D}_{M}^{0}}(\Xi)=D\beta$.

G) $\mathcal{U}\rightarrow\mathfrak{Met}M$ admits a $\mathcal{D}_{M}^{+}%
$-equivariant $\Xi$-local section if and only if there exists a $\mathcal{D}%
_{M}^{+}$-invariant local $1$-form $\beta\in\Omega_{\mathrm{loc}}%
^{1}(\mathfrak{Met}M)^{\mathcal{D}_{M}^{+}}$ such that $\mathrm{hol}_{\phi
}^{\Xi}(\gamma)=\int_{\gamma}\beta\operatorname{mod}\mathbb{Z}$ for any
$\gamma\in\mathcal{C}^{\phi}$.
\end{proposition}

When this result is aplied to the anomaly bundle it provides necessary and
sufficient conditions for physical anomaly cancellation.

For gravitational anomalies the condition P) for physical perturbative anomaly
cancellation can be simplified. It is a classical result of Gel'fand and Fuks
(see \cite{GF}) that the Lie algebra cohomology $H^{k}(\mathfrak{X}(M))=0$\ is
trivial for $k\leq n$. In particular we have $H^{1}(\mathfrak{X}(M))=0$ and by
Proposition \ref{LieAlgCoho} the map $H_{\mathcal{D}_{M}^{0},\mathrm{loc}}%
^{2}(\mathfrak{Met}M)\rightarrow H_{\mathrm{loc}}^{2}(\mathfrak{Met}%
M)^{\mathcal{D}_{M}^{0}}$ is injective. Hence the condition P) in Proposition
\ref{PysAno} is equivalent to the condition

P$^{\prime}$) $\mathcal{U}\rightarrow\mathfrak{Met}M$ admits a $\mathcal{D}%
_{M}^{0}$-equivariant $\Xi$-local section if and only if there exists a
$\mathcal{D}_{M}^{0}$-invariant local $1$-form $\beta\in\Omega_{\mathrm{loc}%
}^{1}(\mathfrak{Met}M)^{\mathcal{D}_{M}^{0}}$ such that $\mathrm{curv}%
(\Xi)=d\beta$.

\subsection{Locality and Universality\label{SubSecLocUniv}}

By using the isomorphism (\ref{Iso}) the condition P$^{\prime}$) can be
expresed in terms of functional forms on $\mathbb{R}^{n}$. Hence the condition
for perturbative gravitational anomaly depends only on the dimension of $M$
and universality is a consequence of locality. Furthermore, the perturbative
anomaly in dimension $n$ can be identified with a cohomology class on
$H^{2}(\mathcal{F}^{\bullet}(J^{\infty}\mathcal{M}_{\mathbb{R}^{n}%
})^{\mathcal{D}_{\mathbb{R}^{n}}^{+}})$.

For determinant bundles $\mathrm{curv}(\Xi)$ is given by the Atiyah-Singer
theorem as a combiantion of characteristic classes, i.e., we have
$\mathrm{curv}(\Xi)=\sigma^{p_{D}}$ for a polynomial $p_{D}\in I^{1+n/2}%
(O(n))$. For example, for twisted Dirac operators associated to a
representation of the Spin group $p_{D}=[\widehat{A}\cdot\mathrm{Ch}]_{1+n/2}%
$, where $\widehat{A}$ is the Dirac genus and $\mathrm{Ch}$\ the Chern
character of the representation and we take the component of degree $1+n/2$.
By Proposition \ref{injective}\ the condition P$^{\prime}$)\ cannot be
satisfied if $p_{D}\neq0$. Hence it is imposible to cancel the perturbative
anomaly with a local counterterm if $p_{D}\neq0$.

If the condition P$^{\prime}$) is satisfied, then the perturbative anomaly
cancells. However, we can still have global anomalies because a $\mathcal{D}%
_{M}^{0}$-equivariant section could be not $\mathcal{D}_{M}^{+}$-equivariant.
We can define a new connection $\Xi^{\prime}=\Xi-\pi^{\ast}\beta$ which is
$\mathcal{D}_{M}^{0}$-invariant by Corollary \ref{Corollarylocal} and it is
$\mathcal{D}_{M}^{+}$-flat by Proposition \ref{LieAlgCoho}. Hence we can
assume that the connection $\Xi$ is $\mathcal{D}_{M}^{+}$-flat. It is shown in
\cite{AnomaliesG} that if a connection $\Xi$ is $\mathcal{D}_{M}^{+}$-flat
then the equivariant holonomy $\mathrm{hol}_{\phi}^{\Xi}(\gamma)$ does not
depend on $\gamma\in\mathcal{C}^{\phi}$ and defines a group homomorphism
$\kappa^{\Xi}\in\mathrm{Hom}(\mathcal{D}_{M}^{+}/\mathcal{D}_{M}%
^{0},\mathbb{R}/\mathbb{Z})$ by $\kappa_{\phi}^{\Xi}=\mathrm{hol}_{\phi}^{\Xi
}(\gamma)$ for any $\gamma\in\mathcal{C}^{\phi}$. We call $\kappa^{\Xi}$\ the
$\mathcal{D}_{M}^{+}$-flat holonomy. Furthermore, condition G) is equivalent
to the existence of a form $\beta\in\Omega_{\mathrm{loc}}^{1}(\mathfrak{Met}%
M)^{\mathcal{D}_{M}^{+}}$ such that $D\beta=0$ and $\kappa_{\phi}^{\Xi}%
=\int_{\gamma}\beta\operatorname{mod}\mathbb{Z}$ for any $\phi\in
\mathcal{D}_{M}^{+}$ and $\gamma\in\mathcal{C}^{\phi}$. Again the possible
forms $\beta$ satisfiying these conditions can be determined by using the
isomorphism (\ref{Iso}).

By Theorem \ref{H1} if $n\neq3\operatorname{mod}4$ then any $\beta\in
\Omega_{\mathrm{loc}}^{1}(\mathfrak{Met}M)^{\mathcal{D}_{M}^{+}}$ satisfiying
$d\beta=0$ is of the form $\beta=d\Lambda$, with $\Lambda\in\Omega
_{\mathrm{loc}}^{0}(\mathfrak{Met}M)^{\mathcal{D}_{M}^{+}}$, and hence
$\int_{\gamma}\beta=\Lambda(\gamma(1))-\Lambda(\gamma(0))=\Lambda(\phi
\cdot\gamma(0))-\Lambda(\gamma(0))=0$ for any $\phi\in\mathcal{D}_{M}^{+}$ and
$\gamma\in\mathcal{C}^{\phi}$. Hence if $n\neq3\operatorname{mod}4$ it is
impossible to cancell the anomally if $\kappa^{\Xi}\neq0$.

But if $n=3\operatorname{mod}4$ then it could be possible to cancel the
anomaly. By Theorem \ref{H1}\ in this case any $\beta\in\Omega_{\mathrm{loc}%
}^{1}(\mathfrak{Met}M)^{\mathcal{D}_{M}^{+}}$ satisfiying $d\beta=0$ is of the
form $\beta=\sigma^{p}+d\Lambda$, with $\Lambda\in\Omega_{\mathrm{loc}}%
^{0}(\mathfrak{Met}M)^{\mathcal{D}_{M}^{+}}$ and $q\in I^{1+n/2}(O(n))$. For
any $\phi\in\mathcal{D}_{M}^{+}$ and $\gamma\in\mathcal{C}^{\phi}$ by Theorem
\ref{MappingTorus} we have
\[
\int_{\gamma}\sigma^{p}=\int_{\gamma}d(CS^{p})=CS_{p,A_{0}}(\phi\cdot
\gamma(0))-CS^{p}(\gamma(0))=\delta_{\phi}^{p}=p(M_{\phi}).
\]
In this way we obtain our main result

\begin{theorem}
\label{PysAnoFlat} Let $\mathcal{U}\rightarrow\mathfrak{Met}M$ be a
$\mathcal{D}_{M}^{+}$-equivariant $U(1)$-bundle with a $\mathcal{D}_{M}^{+}%
$-flat connection $\Xi$. Then $\mathcal{U}$ admits a $\mathcal{D}_{M}^{+}%
$-equivariant $\Xi$-local section if and only if there exists $p\in
I^{(n+1)/2}(O(n))$ such that $\kappa_{\phi}^{\Xi}=p(M_{\phi}%
)\operatorname{mod}\mathbb{Z}$ for any $\phi\in\mathcal{D}_{M}^{+}$.
\end{theorem}

For determinant or Pffafian bundles associated to a family of operators
$D_{g}$ Witten formula has been interpreted as a computation of \ the holonomy
of the Bismut-Freed connection on the quotient determinant bundle
$\mathcal{U}/\mathcal{D}_{M}^{+}\rightarrow\mathfrak{Met}M/\mathcal{D}_{M}%
^{+}$ (e.g. see \cite{BF2}, \cite{DaiFreed}, \cite{freed}). Furthermore, it
can be also interpreted as the equivariant holonomy of the Bismut-Freed
connection on$\ \mathcal{U}\rightarrow\mathfrak{Met}M$. Hence, for
$\mathcal{D}_{M}^{+}$-flat connetions we have $\kappa_{\phi}^{\Xi}=\frac{1}%
{4}\eta_{D}(M_{\phi})$ where $D$ is an operator on the mapping torus. In those
cases the necessary and sufficient condition for gravitational anomaly
cancellation is the existence of $p\in I^{(n+1)/2}(O(n))$ such that
\begin{equation}
\frac{1}{4}\eta_{D}(M_{\phi})=p(M_{\phi})\operatorname{mod}\mathbb{Z}\text{
for any }\phi\in\mathcal{D}_{M}^{+}. \label{AnnomFinal}%
\end{equation}

Hence we show that the unique way to cancell the gravitational anomaly is by
using a Chern-Simons counterterm, as anounced in Section \ref{SectionCS}. If
condition (\ref{AnnomFinal}) is satisfied, then the gravitational anomaly can
be cancelled by using the Chern-Simons action associated to $p$.

We have shown that the posible counterterms necessary to cancell the anomaly
are given by a local $1$-form $\beta\in\Omega_{\mathrm{loc}}^{1}%
(\mathfrak{Met}M)^{\mathcal{D}_{M}^{+}}\simeq\mathcal{F}^{1}(J^{\infty
}\mathcal{M}_{\mathbb{R}^{n}})^{\mathcal{D}_{\mathbb{R}^{n}}^{+}}$, and hence
they are universal. However, the mapping class group $\Gamma_{M}%
=\mathcal{D}_{M}^{+}/\mathcal{D}_{M}^{0}$ and the $\mathcal{D}_{M}^{+}$-flat
holonomy $\kappa^{\Xi}$ are nonlocal objects and can be different for
different manifolds of the same dimension. Hence the condition
(\ref{AnnomFinal}) could be satisfied only for certain manifolds of dimension
$n$ but not for all of them.$\ $We conclude that for global anomalies
universality is not a consequence of locality and should be imposed. One way
to solve the condition (\ref{AnnomFinal}) in a universal way is if%
\begin{equation}
\frac{1}{4}\eta_{D}(N)=p(N)\operatorname{mod}\mathbb{Z}\text{ for any oriented
manifold }N\text{\ of dimension }n+1. \label{AnnomFinalU}%
\end{equation}
We recall that in general $\eta_{D}(N)$ cannot be obtained as the integral of
a local form on $N$. However, in certain cases (e.g. see \cite{Stolz})
$\eta_{D}(N)$ can be expressed as a characteristic number of $N$, and the
condition (\ref{AnnomFinal}) is satisfied.

We recall that Witten has proposed in \cite{WittenFermionic} a stronger
condition for anomaly cancellation. If the condition (\ref{AnnomFinal}) is
satisfied for any $M$ the partition function for any $M$ is determined up to a
phase, but the phases for different $M^{\prime}$s should also be fixed. To
solve this problem Witten proposes as a gereralization of the condition for
anomaly cancellation the condition
\begin{equation}
\frac{1}{4}\eta_{D}(N)=0\operatorname{mod}\mathbb{Z}\text{ for any oriented
manifold }N\text{ of dimension }n+1. \label{AnnomFinalW}%
\end{equation}

In dimension $n\neq3\operatorname{mod}4$ the conditions (\ref{AnnomFinalU})
and (\ref{AnnomFinalW}) are the same. However, if $n=3\operatorname{mod}4$ the
condition (\ref{AnnomFinalU}) can be weaker than (\ref{AnnomFinalW}). For
example, for Majorana fermions in an oriented manifold of dimension 4 we have
(see \cite{WittenFermionic}) $\frac{1}{4}\eta_{D}(N)=\frac{1}{32}%
\sigma(N)\operatorname{mod}\mathbb{Z}$, where $\sigma(N)$ is the signature of
$N$. For $N$ a $K3$ surface we have $\frac{1}{4}\eta_{D}(N)=\frac{1}{2}%
\neq0\operatorname{mod}\mathbb{Z}$ and the\ condition (\ref{AnnomFinalW}) is
not satisfied. However, it follows from the Hirzebruch signature Theorem that
$\sigma(N)=\frac{1}{3}p_{1}(N)$. Hence we have $\frac{1}{4}\eta_{D}%
(N)=\frac{1}{96}p_{1}(N)$ for any oriented $N$, and the anomaly can be
cancelled in the sense of condition (\ref{AnnomFinalW}).

We recall the argument in \cite{WittenFermionic} that leads to the condition
(\ref{AnnomFinalW}). If $X$ is a manifold with boundary such that $\partial
X=M$, and the metric and all the structures can be extended to $X$, Witten
proposes to define the partiton function by $\mathcal{Z}_{D}(g)=\exp(2\pi
i\frac{1}{4}\eta_{D}(X))$. That $\mathcal{Z}_{D}(g)$ is independent of the
manifold $X$\ chosen follows from the Dai-Freed theorem. This argument can be
generalized if condition (\ref{AnnomFinalU})\ is satisfied by defining
$\mathcal{Z}_{D}(g)=\exp(2\pi i(\frac{1}{4}\eta_{D}(X)-\int_{X}p(\omega
^{\overline{g}})))$, where $\overline{g}$ is the extension of $g$ to $X$.

\subsection{Orientation reversing diffeomorphisms and unorientable manifolds}

We have shown in Section \ref{CSnoOrient} that the condition for anomaly
cancellation can be extended to orientation reversing diffeomorphisms. The
condition is%

\begin{equation}
\tfrac{1}{4}\eta_{D}(M_{\phi})=\frac{1}{2}p(M_{\phi^{2}})\operatorname{mod}%
\mathbb{Z}\text{ for any }\phi\in\mathcal{D}_{M}. \label{Anomaly2}%
\end{equation}

We can find a universal version of this condition in the following way. The
mapping torus $M_{\phi}$ can also be obtained as que quotient $(M\times
\mathbb{R})/\mathbb{Z}$ where $\mathbb{Z}$ acts on $M\times\mathbb{R}$ by
setting $\ n\cdot(x,t)=(\phi^{n}(x),t+n)$. Then the projection $M_{\phi^{2}%
}\rightarrow M_{\phi}$ $[(x,t)]_{\phi^{2}}\mapsto\lbrack(\phi(x),2t)]_{\phi}$
is a double cover. If $\phi$ reverses the orientation $M_{\phi}$ is
unorientable, but $M_{\phi^{2}}$ is orientable because $\phi^{2}\in
\mathcal{D}_{M}^{+}$. For any unorientable manifold we have a doble cover
$\widetilde{N}\rightarrow N$ with $\widetilde{N}$ orientable and we have
$\widetilde{M_{\phi}}=M_{\phi^{2}}$. Hence, one way to generalize the
condition (\ref{Anomaly2}) is the following condition
\begin{equation}
\frac{1}{4}\eta_{D}(N)=\frac{1}{2}p(\widetilde{N})\operatorname{mod}%
\mathbb{Z}\ \text{for any manifold }N\text{ of dimension }n+1\text{.}
\label{AnnomFinalU2}%
\end{equation}

Again, as the rigth side is independent of the orientation on $\widetilde{N}$,
this condition can be satisfied only if $\frac{1}{4}\eta_{D}(N)=0,\frac{1}%
{2}\operatorname{mod}\mathbb{Z}$.

The analogous condition proposed in \cite{WittenFermionic} is
\begin{equation}
\frac{1}{4}\eta_{D}(N)=0\operatorname{mod}\mathbb{Z}\ \text{for any manifold
}N\text{ of dimension }n+1\text{.} \label{AnnomFinalW2}%
\end{equation}

Let us consider again the case of Majorana Fermions in dimension 3 studied in
\cite{WittenFermionic}. We have shown before that if we consider only oriented
manifolds the anomaly can be cancelled. But if we admit unorientable manifolds
the situation is different. For $N=\mathbb{RP}^{4}$ we have $\frac{1}{4}%
\eta_{D}(\mathbb{RP}^{4})=\frac{1}{16}\operatorname{mod}\mathbb{Z}$ and hence
the anomaly cannot be cancelled in the sense of the condition
(\ref{AnnomFinalU2}).

One basic question in solid state physics is the problem of how many Majorana
fermions shoud we have in order to have an anomaly free theory. This problem
appears in the analysis of the boundary of a topological superconductor (see
\cite{WittenFermionic} for details). If we consider $\nu$ Majorana fermions,
then we have $\kappa_{\phi}^{\Xi}=\frac{\nu}{4}\eta_{D}(M_{\phi})$. The
condition $\kappa_{\phi}^{\Xi}=0$ for any $\phi\in\mathcal{D}_{M}$ implies
that $\nu$ should be a multiple of $8$. However the condition
(\ref{AnnomFinalW2}) implies that $\nu$ should be a multiple of $16$ (because
for $N=\mathbb{RP}^{4}$ we have $\frac{1}{4}\eta_{D}(\mathbb{RP}^{4})=\frac
{1}{16}\operatorname{mod}\mathbb{Z)}$, a result that coincides with other
analysis in condensed mather physics.

We show that our condition (\ref{AnnomFinalU2}) also implies that $\nu$ should
be a multiple of $16$. For $N=\mathbb{RP}^{4}$ we have $\sigma(\widetilde
{N})=\sigma(S^{4})=0$. Furthermore, as $I^{2}(O(4))$ is one dimensional, this
implies $p(S^{4})=0$ for any $p\in I^{2}(O(4))$. Hence condition
(\ref{AnnomFinalU2}) is satisfied for $N=\mathbb{RP}^{4}$ if $\frac{\nu}%
{16}=0\operatorname{mod}\mathbb{Z}$ and hence $\nu$ should be a multiple of
$16$. Furthermore it is shown in \cite{WittenFermionic} that $\eta
_{D}(N)=\frac{n}{4}$ for any manifold, and hence if $\nu=16k$ we have
$\frac{\nu}{4}\eta_{D}(N)=4k\frac{n}{4}=0\operatorname{mod}\mathbb{Z}$, and
the anomaly cancells with $p=0$. Hence condition (\ref{AnnomFinalU2}) is
satisfied if and only if $\nu$ is a multiple of $16$.

\end{document}